\documentclass[twocolumn,10pt,prl,aps]{revtex4-2}

\usepackage{graphicx,epsfig,amsmath,amssymb,verbatim,color}
\usepackage{dsfont}
\usepackage{float}
\usepackage{tikz}
\usepackage{pgfplots}
\usepackage{xcolor}
\definecolor{darkred}{rgb}{0.8,0.1,0.1}
\usepackage{soul}
\usepackage{amsmath,mathtools}
\usepackage{enumerate}
\usepackage[T1]{fontenc}
\usepackage[utf8]{inputenc}
\usepackage{pifont}
\usepackage{hyperref}
\usepackage{appendix}
\hypersetup{colorlinks=true,citecolor=blue,linkcolor=blue,filecolor=blue,urlcolor=blue,breaklinks=true}
\usepackage[utf8]{inputenc}  
\usepackage{fontenc}         
\usepackage{colortbl}
\usepackage{pifont}
\definecolor{Gray}{gray}{0.92}
\definecolor{Gray2}{gray}{0.75}
\definecolor{maroon}{cmyk}{0,0.87,0.68,0.32}
\usepackage{booktabs}
\usepackage{makecell}
\usepackage{diagbox}
\usepackage{multirow}

\newtheorem{definition}{Definition}
\newtheorem{proposition}{Proposition}
\newtheorem{lemma}[proposition]{Lemma}

\newtheorem{Example}{Example}
\newtheorem{theorem}[proposition]{Theorem}
\newtheorem{remark}{Remark}

\newenvironment{proof}{\noindent \textit{{Proof.}~}}{\hfill $\square$}

\def\squareforqed{\hbox{\rlap{$\sqcap$}$\sqcup$}}
\def\qed{\ifmmode\squareforqed\else{\unskip\nobreak\hfil
		\penalty50\hskip1em\null\nobreak\hfil\squareforqed
		\parfillskip=0pt\finalhyphendemerits=0\endgraf}\fi}
\def\endenv{\ifmmode\;\else{\unskip\nobreak\hfil
		\penalty50\hskip1em\null\nobreak\hfil\;
		\parfillskip=0pt\finalhyphendemerits=0\endgraf}\fi}

\newcommand{\nc}{\newcommand}
\nc{\rnc}{\renewcommand}
\nc{\bra}[1]{\langle#1|}
\nc{\ket}[1]{|#1\rangle}
\nc{\<}{\langle}
\rnc{\>}{\rangle}
\nc{\ketbra}[2]{|#1\rangle\!\langle#2|}
\nc{\braket}[2]{\langle#1|#2\rangle}
\nc{\braandket}[3]{\langle #1|#2|#3\rangle}
\nc{\proj}[1]{| #1\rangle\!\langle #1 |}
\nc{\avg}[1]{\langle#1\rangle}
\nc{\Rank}{\operatorname{Rank}}
\nc{\smfrac}[2]{\mbox{$\frac{#1}{#2}$}}
\nc{\tr}{\operatorname{Tr}}
\nc{\ox}{\otimes}

\nc{\1}{{\mathds{1}}}
\nc{\supp}{{\operatorname{supp}}}

\newcommand{\floor}[1]{\lfloor #1 \rfloor}

\usepackage{mathrsfs}

\begin{document}

	\title{Erasing, Converting, and Communicating: The Power of Resource-Nongenerating Operations}
	
	\author{Xian Shi}\email[]
	{shixian01@gmail.com}
	\affiliation{College of Information Science and Technology,
		Beijing University of Chemical Technology, Beijing 100029, China}

	%
	
	
	
	\date{\today}
	

\begin{abstract}
	We investigate resource nongenerating operations in both static and dynamical quantum resource theories. For the static scenarios, we derive a sufficient condition for state transformations under resource nongenerating operations. Then we construct a dynamical resource theory where resource nongenerating operations constitute the set of free operations, and we propose an axiomatic approach to quantify the dynamical resource. We further analyze the erasure of the dynamical resources. As applications, we establish bounds on the rate of state conversions under resource nongenerating operations in a generic convex resource theory and obtain capacity bounds for classical communication tasks assisted by dynamical coherence. Our results clarify the key roles of resource nongenerating operations in quantum information processing tasks.
\end{abstract}

			\maketitle

\section{Introduction}

With the rapid development of quantum information science \cite{nielsen2002quantum}, it has been gradually recognized that quantum entanglement \cite{horodecki2009quantum}, quantum discord  \cite{modi2012classical}, quantum coherence \cite{streltsov2017colloquium}, and other related quantities play essential roles in quantum information processing. As more quantum information-theoretic objects emerge, it becomes increasingly valuable to develop a unified framework for understanding the advantages conferred by quantum theory. Quantum resource theory provides such a framework  \cite{chitambar2019quantum}. A generic resource theory $R$
 consists of at least two elements: a set of free states 
$\mathcal{F}_R$ , which do not possess the resource, and a set of free operations 
$\mathcal{O}_R$, which cannot generate the resource from free states. Two fundamental issues in any resource theory are (i) how to quantify the amount of resource in a state, and (ii) how to determine whether one state can be transformed into the other under the allowed operations. However, analyzing state convertibility under 
$\mathcal{O}_R$  is always challenging. For example, in the resource theory of entanglement, the free states are separable states and the free operations are local operations and classical communication (LOCC). Deciding whether two mixed entangled states are interconvertible under LOCC is notoriously difficult. One approach to overcome this difficulty is to enlarge the set of free operations. The maximal such set within the framework of a resource theory is the class of resource-nongenerating operations, which map free states to free states [4]. State convertibility under this broader class of operations has been investigated for entanglement \cite{contreras2019resource}, coherence\cite{aberg2006quantifying}, and others \cite{ku2022,Vempati2022}.

In parallel with quantum states, quantum channels can also be viewed as carriers of resources. Similar to states, channels may evolve under various configurations. The most general transformation between two channels is described by a quantum supermap  \cite{chiribella2008transforming,gour2019comparison}, among which superchannels have been found to play key roles in quantum information processing  \cite{liu2019resource,regula2021,liu2020operational,fang2022no,shi2025bounds,zhu2025classical}. The study of quantum dynamical resource theories has thus attracted growing attention. General frameworks for the quantification and transformation of dynamical resources have been proposed \cite{regula2021one,takagi2020,PhysRevA.99.032317,PhysRevResearch.3.023096,gour2021uniqueness,gour2019quantify,liu2020op}, and dynamical extensions of specific static resource theories have also been explored \cite{wang2019resource,wang2019quantifying,gour2020dynamical,takagi2022one,zhou2022quantifying,hsieh2024resource}. Following the static cases, resource-nongenerating operations (RNGs) also play a central role in the dynamical setting \cite{puchala2021dephasing,kim2021one,hsieh2025dynamical,PhysRevA.111.022447}. Nevertheless, little is studied about the dynamical resources associated with RNGs in the general framework of resource theories.
In this work, we address this gap. Our main contributions are summarized as follows:
\begin{itemize}
	\item State convertibility: We establish a sufficient condition for the convertibility between two states under RNGs.
	\item Axiomatic framework for the dynamical resource: We introduce a generic dynamical resource theory in which the free channels are RNGs, and develop an axiomatic approach to quantify the resource of channels.
	\item Resource erasure: We investigate the task of erasing dynamical resources and provide a bound of this process.
	\item  Applications to quantum information tasks:
	\begin{itemize}
		\item Rates under asymptotic RNGs: We present a lower bound of the conversion rates of quantum states under asymptotic RNGs of a generic convex resource theory.
		\item  Classical communication capacity: We analyze the capacity of a classical communication task assisted by the dynamical resources of maximally incoherent operations.
	\end{itemize}
\end{itemize}

The rest of this article is organized as follows. In Sec. \ref{s1}, we show first the mathematical notions needed in this article, then we present the assumptions  of the quantum static resource theories considered here. Besides, we obtain a sufficient condition on the transformations of two quantum states under the resource nongenerating operations. In Sec. \ref{s3}, we propose how to quantify a channel under the dynamical resource theory of RNGs, then we address the erasure task of a quantum operation with the aid of absolutely resource nongenerating operations. In Sec. \ref{s4}, we consider two applications of the RNGs, we first present the bounds of conversion rates of quantum states through the resource nongenerating operations, we also address the classical capacity of a communication task under the quantum dynamical resource of coherence.

\section{Preliminary Knowledge}\label{s1}

 Here we first introduce a generic framework of a resource theory. Assume $\mathcal{H}$ is {\color{black}a Hilbert space associated to the system} with $dim\mathcal{H}=d$, $\mathcal{D}_{\mathcal{H}}$ is a set consisting of the set of states, which are semidefinite positive and $\mathrm{tr}\rho=1,$
\begin{align*}
	\mathcal{D}_{\mathcal{H}}=\{\rho|\rho\ge 0,\mathrm{tr}\rho=1\},
\end{align*}  
$\mathcal{C}_{\mathcal{H}}$ is a set of linear maps which are trace-preserving and completely positive, here we denote the element in $\mathcal{C}_{\mathcal{H}}$ as an operation.
 
A standard static resource theory defined on $\mathcal{H}$ consists of a set of free states $\mathcal{F}_{R}\subset \mathcal{D}_{\mathcal{H}}$ and a set of free operations $\mathcal{O}_{R}\subset \mathcal{C}_{\mathcal{H}}$, which is denoted as $\langle \mathcal{F}_{R},\mathcal{O}_{R}\rangle$ \cite{chitambar2019quantum}. {\color{black} For a resource $\langle \mathcal{F}_R,\mathcal{O}_R\rangle$, there exists a fundamental restriction: for any $\Phi\in \mathcal{O}_R$ and $\rho\in \mathcal{F}_R$, $\Phi(\rho)\in \mathcal{F}_R$, which is denoted as the \emph{golden rule}. We call the operations compatible with the \emph{golden rule} resource nongenerating operations ({RNGs}), and denote their set by $\mathcal{M}_R$,  defined as,
\begin{align*}
	\mathcal{M}_{R}=\{\mathcal{L}\in \mathcal{C}_{\mathcal{H}}|\mathcal{L}(\rho)\in \mathcal{F}_{R},\forall  \rho\in \mathcal{F}_R\}.
\end{align*}
Based on the characterization of $\mathcal{M}_R,$ any class of free operations is a subset of $\mathcal{M}_R,$ that is,
comparing with a generic class of  operations $\mathcal{O}_R$, the constraints defining resource nongenerating operations (RNGs) are typically less restrictive. For instance, in the resource theory of bipartite entanglement, non-entangling operations coincide with RNGs \cite{brandao2008}, and this class is strictly larger than local operations and classical communication (LOCC) \cite{horodecki2009quantum,chitambar2014everything}. Moreover, since RNGs form the maximal set of operations compatible with the \emph{golden rule} of resource theories, they can be used to reveal or rule out certain features of a given theory-most notably reversibility. In the resource theory of coherence, coherence is irreversible under incoherent operations \cite{PhysRevLett.116.120404}, whereas it becomes reversible under maximally incoherent operations \cite{Chitambar2018dephasing,8863412}. Similar phenomena have been reported in resource-theoretic approaches to thermodynamics \cite{horodecki2013fundamental,brandao2015second,Faith2019thermodynamic}. Nevertheless, the irreversibility of entanglement persists not only under LOCC \cite{PhysRevLett.86.5803} and under operations that preserve positivity of the partial transpose \cite{PhysRevLett.119.180506}, but also under non-entangling operations \cite{lami2023}. At last, this class of operations often enable stronger no-go results than weaker free sets of operations.  For example, in the resource of entanglement, as LOCC is difficult to characterize, considerable effort has been devoted to the study of state convertibility under RNGs \cite{brandao2008,contreras2019resource,chitambar2020entanglement,lami2023}.
}

 Next we list the properties that the static convex resource theories have considered here:
 \begin{itemize}
 	\item[(R1)]\label{R1} $\mathcal{F}_{R}$ is convex {\color{black}and compact}.
 	 	\item[{\color{black}(R2)}] $\mathcal{M}_{R}$ is convex.
 	\item[(R3)]  if $\mathcal{E}_1,\mathcal{E}_2\in\mathcal{M}_R$, $\mathcal{E}_1\circ\mathcal{E}_2,\mathrm{tr}_S\mathcal{E}_1\in\mathcal{M}_R$, here $S$ stands for the global systems and partial systems.
 	\item[(R4)] The identity operation $\mathcal{I}\in\mathcal{M}_R$. 
  
 \end{itemize}
 
 \begin{remark}

 		Assume $\Lambda_1$ and $\Lambda_2$ are two operations in $\mathcal{M}_R$, $\rho$ is an arbitrary state in $\mathcal{F}_R$, when $p\in [0,1]$,
 		\begin{align*}
 			(p\Lambda_1+(1-p)\Lambda_2)(\rho)=	p\Lambda_1(\rho)+(1-p)\Lambda_2(\rho)\in \mathcal{F}_R,
 		\end{align*}
 		the first equality is due to that $\mathcal{F}_R$ is convex. Hence, (R1) $\Longrightarrow$ {\color{black}(R2).}

 \end{remark}
 
For a set of free states $\mathcal{F}_{R}$, it is not always closed under operation $\otimes$. For example, when the resource is the set of genuinely entangled states in a tripartite system, the free set of states is the biseparable states, which is not closed under the tensor operation \cite{yamasaki2022activation,Carlos2022}. In response to the above problem, we define the following set of states $\mathcal{F}_{R}^{\tilde{T}}$ as follows,
\begin{align*}
	\mathcal{F}_{R}^{\tilde{T}}=\{\sigma\in\mathcal{F}_{R}|\rho\otimes\sigma\in\mathcal{F}_{R},\forall\rho\in\mathcal{F}_{R}\}.
\end{align*}
Generally, the above set is not always empty. For the resource theory of genuinely entangled states, the maximally mixed state $\frac{I_A}{d}\otimes\frac{I_B}{d}\otimes\frac{I_C}{d}$ is in $\mathcal{F}_{R}^{\tilde{T}}$.

A key tool to study generic resource theories is the function which quantifies the resource of a state. {\color{black}It is a family of functions which map the states to $\mathbb{R}^{+}\cup\{0\}$.} We denote the above functions as quantifiers for a given resource theory. Various kinds of quantifiers have been proposed for generic resource theories. Generally, a valid quantifier $\mathtt{f}$ for static resource theories should satisfy the following properties O1 and O2 at least,
\begin{itemize}
	\item[O1.] $\mathtt{f}(g)=0,$ $\forall g\in \mathcal{F}_R$.
	\item[O2.] $\mathtt{f}(\Lambda(\rho))\le \mathtt{f}(\rho)$, $\forall \Lambda\in \mathcal{M}_{R}$.
	\item[O3.] $\mathtt{f}(\rho)=0\Leftrightarrow \rho\in \mathcal{F}_R.$
	\item[O4.] For arbitrary states $\rho_1$ and $\rho_2$, $p\in (0,1)$,
	\begin{align*}
		\mathtt{f}(p\rho_1+(1-p)\rho_2)\le p\mathtt{f}(\rho_1)+(1-p)\mathtt{f}(\rho_2).
	\end{align*}
\end{itemize}

When the quantifier $\mathtt{f}$ satifies the property O3, we say $\mathtt{f}$ is faithful. And $\mathtt{f}$ is convex if the quantifier $\mathtt{f}$ satisfies the property O4. 
\begin{remark}
	{\color{black} As $\mathcal{M}_R$ is the maximal, if a quantifier $\mathtt{f}(\cdot)$ satisfies (O2) under $\langle\mathcal{F}_R,\mathcal{M}_R\rangle$, then $$\mathtt{f}(\Lambda(\rho))\le \mathtt{f}(\rho),\forall \Lambda\in\mathcal{O}_{R}.$$ 	
	That is, if we could show $\mathtt{f}$ is a quantifier for $\langle\mathcal{F}_R,\mathcal{M}_R\rangle$, the same is for $\langle\mathcal{F}_R,\mathcal{O}_R\rangle.$}
\end{remark}

A common-used quantifier is the robustness-based measure. It is widely applied in the specific resource theories, such as, entanglement \cite{vidal1999robustness}, coherence \cite{napoli2016robustness} and quantum steering \cite{piani2015}. Here we propose the robustness-based measures of generic resoure theory. Assume $\rho\in\mathcal{D}_{\mathcal{H}},$ its robustness-based measure $R_G(\rho)$ and standard robustness-based measure $\mathscr{R}_G(\rho)$
are defined as follows,
\begin{align}
{R}_G(\rho)=&\inf \{r|\frac{\rho+r\sigma}{1+r}\in \mathcal{F}_{R},\sigma\in \mathcal{D}\},\nonumber\\
	\mathscr{R}_G(\rho)=&\inf \{r|\frac{\rho+r\sigma}{1+r}\in \mathcal{F}_{R},\sigma\in \mathcal{F}_{R}\}.\label{r}
\end{align}

\begin{remark}
	Comparing the definition of $\mathscr{R}_{G}(\cdot)$ with $R_{G}(\cdot)$, the difference is that the mixing term in $\mathscr{R}_{G}(\cdot)$ is a free state, while it is an arbitrary state for $R_G(\cdot)$. When the resource is the bipartite entanglement, it was shown that the values of the robustness of entanglement and the generalized robustness of entanglement for a pure state are the same \cite{steiner2003generalized}.  Nevertheless, it maybe not valid for generic resource theories.
\end{remark}

\begin{theorem}\label{th5}
	Assume $\langle\mathcal{F}_{R},\mathcal{M}_R\rangle$ satisfies (R1)-(R4), then the quantifiers $R_G(\cdot)$ and $\mathscr{R}_G(\cdot)$ for $\langle\mathcal{F}_R,\mathcal{M}_R\rangle$ satisfy the following properties,
	\begin{itemize}
		\item[(i)] For any mixed state $\rho$, $R_G(\rho),\mathscr{R}_G(\rho)\in[0,\infty).$
		\item[(ii)] $R_G(\cdot)$ and $\mathscr{R}_G(\cdot)$ are valid quantifiers.
		\item[(iii)] $R_G(\cdot)$ and $\mathscr{R}_G(\cdot)$ satisfy the properties O3 and O4.
		\item[(iv)] Assume $\rho$ is a state, $\sigma$ is the optimal free state in terms of $\mathcal{R}_G(\cdot)$ for $\rho$.  If $s>\mathcal{R}_G(\rho)$,
		\begin{align*}
			\frac{\rho+s\sigma}{1+s}\in\mathcal{F}_{R}.
		\end{align*}
		Here $\mathcal{R}_G(\cdot)$ is $R_G(\cdot)$ or $\mathscr{R}_G(\cdot).$
	\end{itemize}
\end{theorem}
The proof of Theorem \ref{th5} is placed in the Sec. \ref{app}.

Another type of measure for generic resource is defined in terms of a geometric method,
\begin{align*}
	G_R(\ket{\psi})=\inf_{\sigma\in \mathcal{F}_R}[1-F(\ket{\psi}\bra{\psi},\sigma)],
\end{align*}
where $F(\rho,\gamma)=||\sqrt{\rho}\sqrt{\gamma}||_1$ and the infimum takes over all the states $\sigma\in\mathcal{F}_R.$ This measure was first studied for pure states in entanglement theory \cite{SHIMONY1995,Barnum_2001}. It was also applied to study the resource of multipartite entanglement \cite{wei2003,chen2014} and coherence in \cite{Zhang_2017}.

A high-profile research subject for the generic resource theory is to study the feasibility of the transformations between two states under the given resource theory. Next we present a sufficient condition when the convertibility between a pure state and a generic state under the RNGs.
\begin{theorem}\label{t1}
	Assume $\langle\mathcal{F}_{R},\mathcal{M}_{R}\rangle$ is a convex resource, $\psi$ is a pure state with resources. If $\frac{1}{1+\mathcal{R}_G(\sigma)}+G_R(\ket{\psi})\ge1$, then there exists an operation $\Lambda\in\mathcal{M}_{R}$ such that $\Lambda(\psi)=\sigma.$  {\color{black}Here $\mathcal{R}_G(\cdot)$ is $R_G(\cdot)$ or $\mathscr{R}_G(\cdot).$}
\end{theorem}
\begin{proof}
	Assume $\rho\in \mathcal{D}_{\mathcal{H}},$ then denote $\Lambda(\cdot)$ as the following map,
	\begin{align*}
		\Lambda(\rho	)=\mathrm{tr}(\psi\rho)\sigma+(1-\mathrm{tr}(\psi\rho))\delta,
	\end{align*}
	here we assume $\delta$ is the optimal for $\sigma$ in terms of the definition $\mathcal{R}_G(\cdot)$ in $(\ref{r})$, $i.$ $e.$, $\sigma+{\mathcal{R}}_G(\sigma)\delta\in (1+\mathcal{R}_G(\sigma))\mathcal{F}_{R}$. Next we show that $\Lambda(\cdot)$ is in $\mathcal{M}_{R}$. Assume $\rho\in\mathcal{F}_{R}$,  then 
	\begin{align*}
		\Lambda(\rho)=&\mathrm{tr}(\psi\rho)\sigma+(1-\mathrm{tr}(\psi\rho))\delta\\
		=&\mathrm{tr}(\psi\rho)(\sigma+\frac{(1-\mathrm{tr}(\psi\rho))}{\mathrm{tr}(\psi\rho)}\delta).
	\end{align*}
	As $$\frac{1}{1+\mathcal{R}_G(\sigma)}\ge{\color{black}\max\limits_{\rho\in\mathcal{F}_R} F(\psi,\rho),}$$ then based on Theorem \ref{th5} (iv), 
	$\Lambda(\rho)\in\mathcal{F}_{R}$. Hence, we construct a RNG that turns $\psi$ to $\sigma$.
\end{proof}

\hspace{3mm}

\section{Dynamical Resource Theory}\label{s3}
Before presenting the results on the transformations of the resource nongenerating operations, we recall the structures of such maps, which are denoted as superchannels \cite{chiribella2008transforming,gour2019comparison},
\begin{align}
M:\hspace{3mm}& \mathcal{C}_{\mathcal{H}_A}\longrightarrow\mathcal{C}_{\mathcal{H}_A}\nonumber\\
&	\mathcal{E}\hspace{4mm}\longrightarrow \mathcal{M}\circ(\mathcal{E}\otimes\mathcal{I}_B)\circ\mathcal{N},\label{f1}
\end{align}
here $\mathcal{N}$ and $\mathcal{M}$ are some quantum channels from $\mathcal{D}_{\mathcal{H}_A}$ to $\mathcal{D}_{\mathcal{H}_{AB}}$ and from $\mathcal{D}_{\mathcal{H}_{AB}}$ to $\mathcal{D}_{\mathcal{H}_A}$, respectively, $\mathcal{I}_B$ is the identity on some ancillary system $\mathcal{H}_B$. A intuitive thought to consider RNGs is to study all the supermaps (\ref{f1}) which could transform any RNG to the other RNG. However, on one hand, it is hard to characterize all such superchannels for a generic resource theory, on the other side, these supermaps lack clear physical intepretations generally. Therefore, in this manuscript, we adopt a more direct scheme which makes the representations of superchannels own clear physical meanings.

Here we restrict the pre-processing operation $\mathcal{N}$ and post-processing operation $\mathcal{M}$ to be RNGs. This restriction is due to the expection that RNGs should be generated freely. Similar to the analysis in the static resource theories, 
the set $\mathcal{M}_{R}$ is not always closed under the tensor product. To further study the properties of  the dynamical resource, we introduce the following concept,
\begin{definition}\label{d1}
	Assume $\Gamma\in \mathcal{M}_{R}$, if $\Gamma$ satisfies the following property, 
	\begin{align*}
		\Lambda\otimes\Gamma\in \mathcal{M}_{R},\hspace{3mm}\forall \Lambda\in \mathcal{M}_R,
	\end{align*}
	then we say $\Gamma$ is an absolutely resource nongenerating operation (ARNG). The set of all ARNGs is denoted as 	$\mathcal{M}_{R}^{\tilde{T}}$,
	\begin{align*}
		\mathcal{M}_{R}^{\tilde{T}}=\{\Gamma\in \mathcal{M}_R|\Lambda\otimes\Gamma\in\mathcal{M}_{R},\forall\Lambda\in\mathcal{M}_{R}\}.
	\end{align*}
\end{definition}

\begin{remark}\label{re4}
{\color{black} As $\mathcal{I}\in \mathcal{M}_R,$ when $\mathcal{E}$ is an arbitrary element in $\mathcal{M}^{\tilde{T}}_R$, then $\mathcal{I}\otimes\mathcal{E}\in \mathcal{M}_R$.
	
	Next for a resource theory $\langle\mathcal{F}_R,\mathcal{M}_R\rangle$, we say $\Lambda\in \mathcal{M}_R$ is completely RNG if $\mathcal{I}_{anc}\otimes\Lambda\in \mathcal{M}_R,$ here $\mathcal{I}_{anc}$ is the identity map of an arbitrary ancillary system $\mathcal{H}_{anc}.$ let us denote the set of all completely RNGs, ${\mathcal{CM}_R},$ as 
	\begin{align*}
		{\mathcal{CM}_R}=\{\Gamma\in \mathcal{M}_R|\mathcal{I}\otimes\Gamma\in\mathcal{M}_R\},
	\end{align*}if $\mathcal{M}_R=\mathcal{M}_R^{\tilde{T}},$ based on the above paragraph, if $\Lambda\in \mathcal{M}_R$, then $\mathcal{I}\otimes\Lambda\in\mathcal{M}_R,$ that is, $\mathcal{M}_R\subset\mathcal{CM}_R.$ And as $\mathcal{CM}_R\subset\mathcal{M}_R$. Hence, $\mathcal{CM}_R=\mathcal{M}_R.$
	
	At last, for a resource theory $\langle\mathcal{F}_R,\mathcal{M}_R\rangle$, let us denote $\overline{\mathcal{CM}_R}$ as 
	\begin{align*}
		\overline{\mathcal{CM}_R}=\{\Gamma\in \mathcal{M}_R|\mathcal{I}\otimes\Gamma,\Gamma\otimes\mathcal{I}\in \mathcal{M}_R\},
	\end{align*}
	here $\mathcal{I}$ is the identity. If $\mathcal{E}$ and $\mathcal{F}$ are arbitrary two elements in $\overline{\mathcal{CM}_R}$, as ${\mathcal{M}_R}$ is closed under the composition, then 
	\begin{align}
		(\mathcal{E}\otimes\mathcal{I})\circ(\mathcal{I}\otimes\mathcal{F})=\mathcal{E}\otimes\mathcal{F}\in{ \mathcal{M}_R}.\label{re1}
	\end{align}
	If $\overline{\mathcal{CM}_R}=\mathcal{M}_R$, then $\mathcal{M}_R^{\tilde{T}}\subset \mathcal{M}_R=\overline{\mathcal{CM}_R}$; due to (\ref{re1}), $\mathcal{M}_R=\overline{\mathcal{CM}_R}\subset \mathcal{M}_R^{\tilde{T}}$, then $\mathcal{M}_R=\mathcal{M}_R^{\tilde{T}}$.

}
\end{remark}

Here we would generalize the operations on the ancillary system to any operation in $\mathcal{M}_R^{\tilde{T}}$.
\begin{definition}[Free super-channel of RNGs]
Assume $\mathcal{G}:\mathcal{M}_{R}\rightarrow\mathcal{M}_{R}$ is a superchannel, if $\mathcal{G}$ can be written as the following,
\begin{align*}
	\mathcal{G}(\mathcal{E})=\mathcal{N}_2\circ(\mathcal{E}\otimes\Lambda)\circ\mathcal{N}_1,
\end{align*}
here $\mathcal{N}_1$ and $\mathcal{N}_2$ are RNGs, and $\Lambda\in \mathcal{M}_{R}^{\tilde{T}}$. In this manuscript, we denote the set of free super-channel of RNGs as $\mathcal{SC}_{\mathcal{H}}.$
\end{definition}

Similar to the static resource, it is also necessary to employ the functions to quantify the dynamical resources. Assume $\mathcal{C}_{\mathcal{H}}$ is the set of all channels acting on the states $\mathcal{D}_{\mathcal{H}},$ we say $\mathbb{F}:\mathcal{C}_{\mathcal{H}}\rightarrow \mathbb{R}^{+}\cup\{0\}$ is a RNG monotone if $\mathbb{F}(\cdot)$ satisfies properties ({P1}) and ({P2}):
\begin{itemize}
	\item[(P1).]  $\mathbb{F}(\Gamma)=0,$ $\forall \Gamma\in \mathcal{M}_{R}$.
	\item[(P2).] Assume $\mathcal{E}$ is an operation of $\mathcal{H}$, 
	\begin{align*}
		\mathbb{F}(\Lambda(\mathcal{E}))\le \mathbb{F}(\mathcal{E}),\hspace{5mm}\forall \Lambda\in \mathcal{SC}_{\mathcal{H}}.
	\end{align*}
\end{itemize}

\noindent	Besides, a RNG monotone $\mathbb{F}(\cdot)$ is faithful if $\mathbb{F}(\cdot)$ satisfies (P3):
	\begin{itemize}
	\item[(P3).] $\mathbb{F}(\mathcal{G})=0$ if and only if $\mathcal{G}\in \mathcal{M}_{R}$.
\end{itemize}
\noindent	$\mathbb{F}(\cdot)$ is convex, if it satisfies (P4):
\begin{itemize}
	\item[(P4).] $\mathbb{F}(p\mathcal{E}_1+(1-p)\mathcal{E}_2)\le p\mathbb{F}(\mathcal{E}_1)+(1-p)\mathbb{F}(\mathcal{E}_2).$\hspace{7mm} $\forall\mathcal{E}_1,\mathcal{E}_2\in \mathcal{C}_{\mathcal{H}},$ $\forall p\in [0,1].$
\end{itemize}

Assume $D: \mathcal{D}_{\mathcal{H}}\times\mathcal{D}_{\mathcal{H}}\longrightarrow \mathbb{R}$ is a bilinear function, and it satisfies the data-processing property, $i.$ $e.$, if $\rho,\sigma\in \mathcal{D}_{\mathcal{H}}$, $\mathcal{E}\in \mathcal{C}_{\mathcal{H}}$, then
\begin{align*}
	D(\rho,\sigma)\ge D(\mathcal{E}(\rho),\mathcal{E}(\sigma)).
\end{align*}  
 In this manuscript, we also assume $D(\cdot,\cdot)$ is a generalized distance, which satisfies the following properties, 
\begin{itemize}
	\item[(D1).] $D(\rho,\sigma)\ge 0$, $\forall \rho,\sigma\in\mathcal{D}_{\mathcal{H}}$, the equality happens if and only if $\rho=\sigma$.
		\item[(D2).] $D(\rho,\sigma)=D(\sigma,\rho)$, $\forall \rho,\sigma\in\mathcal{D}_{\mathcal{H}}.$
	\item[(D3).] $D(\rho,\sigma)+D(\sigma,\gamma)\ge D(\rho,\gamma)$,  $\forall \rho,\sigma,\gamma\in \mathcal{D}_{\mathcal{H}}$.
\end{itemize}  

\begin{definition}
	Assume $\mathcal{F}_R$ is a set of free states, if $\{\rho_n\}$ is a sequence of states in $\mathcal{F}_{R}$, 
	\begin{align*}
		\lim_{n\rightarrow\infty}D(\rho_n,\sigma)=0\Rightarrow \sigma\in\mathcal{F}_{R},
	\end{align*}
	then $\mathcal{F}_R$ is closed under the distance $D(\cdot,\cdot)$.
	
If $\mathcal{B}$ is a set of channels, $\{\mathcal{E}_k\}$ is arbitrary series of  channels which belong to $\mathcal{B}$, 
\begin{align*}
	\lim_{k\rightarrow\infty}\sup_{\rho}D(\mathcal{E}_k(\rho),\mathcal{G}(\rho))=0\Rightarrow \mathcal{G}\in \mathcal{B},
\end{align*}
  we say $\mathcal{B}$ is closed under $D(\cdot,\cdot)$.
\end{definition}

In this subsection, we will denote $\mathcal{M}_{\mathcal{H}}$ and $\mathcal{M}_{\mathcal{H}}^{\tilde{T}}$ as the set of all RNGs and ARNGs acting on $\mathcal{H},$ respectively. Now we could introduce the following measures induced by $D$ to quantify the RNG,
\begin{align}
	\mathbb{F}_D(\mathcal{E})=\inf_{\Lambda\in \mathcal{M}_{\mathcal{H}_A}}\sup_{\Gamma\in \mathcal{M}_{\mathcal{H}_B}^{\tilde{T}},\rho_{AB}}D[(\mathcal{E}\otimes\Gamma)(\rho_{AB}),(\Lambda\otimes\Gamma)(\rho_{AB})],\label{cd}
\end{align}
here the infimum takes over all the operations $\Lambda$ in $\mathcal{M}_{\mathcal{H}_A}$, and the supermum take over all the states $\rho\in \mathcal{D}_{{AB}}$ and $\Gamma\in \mathcal{M}_{\mathcal{H}_B^{\tilde{T}}}$.
\begin{theorem}\label{th3}
	Assume ${D}$ is a distance between two states, then
		\begin{itemize}
		\item[1.] if ${D}$ satisfies the data-processing property, 
		\begin{align*}
			D(\Lambda(\rho),\Lambda(\sigma))\le D(\rho,\sigma),\hspace{3mm}\forall \rho,\sigma\in \mathcal{D}_{\mathcal{H}}, \Lambda\in \mathcal{C}_{\mathcal{H}}.
		\end{align*}
		Then $\mathbb{F}_{D}(\cdot)$ is a RNG monotone. 
		\item[2.] if $\mathcal{E}_2\in \mathcal{M}_{R}^{\tilde{T}}$, then
		\begin{align*}
			\mathbb{F}_{D}(\mathcal{E}_1\otimes\mathcal{E}_2)=\mathbb{F}_{D}(\mathcal{E}_1)
		\end{align*}
		\item[3.] if $\mathcal{M}_{R}$ is closed under $D(\cdot,\cdot)$, $\mathbb{F}_D(\cdot)$ is faithful.
	\end{itemize}  
\end{theorem}

The proof of Theorem \ref{th3} is placed in Sec. \ref{app}

Before we introduce the next result, we would recall the max-relative entropy between $\rho$ and $\sigma$ \cite{datta2009}, 
\begin{align*}
	D_{\max}(\rho||\sigma)=\log_2\inf\{\lambda>0|\rho\le \lambda\sigma\},
\end{align*}
where the infimum takes over all the positive numbers, and the base of logarithm is 2.  It satisfies the following properties, 
\begin{itemize}
	\item[1.] [faithfulness]   $D_{\max}(\rho||\sigma)=0\Longleftrightarrow \rho=\sigma.$
	\item[2.]  [data-processing] $D_{\max}(\rho||\sigma)\ge D_{\max}(\Lambda(\rho)||\Lambda(\sigma))$, \hspace{5mm} if $\Lambda\in \mathcal{O}_{R}$.
\end{itemize}
Based on Theorem \ref{th3}, we have $\mathbb{F}_{D_{\max}}(\cdot)$ is a RNG monotone, it is also faithful if $\mathcal{M}_{R}$ is closed under $D_{\max}(\cdot||\cdot)$.

Next we introduce the RNG robustness $\mathbb{R}_L(\cdot)$ of a channel $\mathcal{E}$,
\begin{align}
	\mathbb{R}_L(\mathcal{E})=&{\color{black}-\log_2}\mathbb{L}(\mathcal{E})\nonumber\\
	\mathbb{L}(\mathcal{E})=&\sup\{p|p\mathcal{E}+(1-p)\mathcal{G}\in \mathcal{M}_{{R}}\},\label{roc}
\end{align} 
where $\mathcal{G}$ takes over all channels in $\mathcal{C}_{\mathcal{H}}$. Its $\epsilon$-smooth RNG robustness is defined as follows, 
\begin{align*}
	\mathbb{R}_L^{\epsilon}(\mathcal{E})=\inf_{\mathcal{E}^{'}}\mathbb{R}_L(\mathcal{E}^{'}),
\end{align*} 
where the infimum takes over all the channels $\mathcal{E}^{'}$ with property $\frac{1}{2}||\mathcal{E}^{'}-\mathcal{E}||_{\diamond}\le\epsilon.$ {\color{black} Here $||\cdot||_{\diamond}$ is a commonly used metric to consider the distance between two channels, it is defined as follows, $||\mathcal{F}||_{\diamond}=\sup_{k\in \mathbb{N}}\sup_{||\sigma||_1=1}||\mathcal{F}\otimes\mathcal{I}_k(\sigma)||_1,$ and $||\sigma||_1=\mathrm{tr}\sqrt{\sigma^{\dagger}\sigma}$.}
Then we consider the following task,
\begin{definition}
	Assume $\epsilon\in (0,1)$, $\mathcal{E}\in \mathcal{C}_{\mathcal{H}}$,  if there exists an ensemble of pairs of free reversible channels $\{p_i,\mathcal{U}_i,\mathcal{V}_i\}_{i=1}^k$ (\emph{$i$.$e$.} unitary operations $\mathcal{U}_i,\mathcal{V}_i\in \mathcal{M}_{R}$, $\forall i=1,2,\cdots,k$), $\mathcal{G}\in \mathcal{M}_{R}^{\tilde{T}}$ and an $\Lambda\in \mathcal{M}_{R}$ such that
	\begin{align*}
		\frac{1}{2}||\sum_ip_i \mathcal{U}_i\circ(\mathcal{E}\otimes\mathcal{G})\circ\mathcal{V}_i-\Lambda||_{\diamond}\le \epsilon,
	\end{align*} 
 then we say the above is an $\epsilon$-destruction process of RNG for $\mathcal{E}.$
	
		The $\epsilon$-destruction cost for RNG, $C^{\epsilon}(\mathcal{E})$  is defined as 
	\begin{align*}
		C^{\epsilon}(\mathcal{E})=\min\log_2k,
	\end{align*}
	where the minimization takes over all the $\epsilon$-destruction process of RNG for $\mathcal{E}$.
\end{definition}

The above task is widely addressed in quantum static resources, such as, quantum correlation \cite{groisman2005quantum}, bipartite entanglement \cite{berta2018disentanglement}, asymmetry \cite{wakakuwa2017symmetrizing}, and quantum imaginarity \cite{shi2025}.

\begin{theorem}\label{th4}
	Assume $\mathcal{M}_R$ satisfies the following properties, 
	\begin{itemize}
		\item $\mathcal{M}_R$ is closed under the permuation operation. That is, for a system $\mathcal{H}^{\otimes n}$, $\mathcal{U}_{\pi}(\cdot)=U^{\dagger}_{\pi}\cdot U_{\pi}\in \mathcal{M}_R$, $\pi$ is some permutation in $S_n$.
		\item $\mathcal{M}_R=\mathcal{M}_R^{\tilde{T}}.$
	\end{itemize}
\noindent	
When	 $\mathcal{E}$ is a channel, then for arbitrary $\eta$ and $\epsilon$ such that $0<\eta<\epsilon<1$, we have
\begin{widetext}
	\begin{align}
		\mathbb{R}_L^{\sqrt{\epsilon}}(\mathcal{E})\le	C^{\epsilon}(\mathcal{E})\le \min\left\{\log_2\left[\frac{2(2^{\mathbb{R}_L^{\epsilon-\eta}(\mathcal{E})}-1)e^{\frac{1}{6}}}{\eta^2\pi }+\frac{1}{2^{-\mathbb{R}_L^{\epsilon-\eta}(\mathcal{E})}(1-2^{-\mathbb{R}_L^{\epsilon-\eta}(\mathcal{E})})}+1\right],\mathbb{R}_L^{\epsilon-\eta}(\mathcal{E})+\log_2\frac{1}{\eta^2}\right\}.
	\end{align} 
\end{widetext}

\end{theorem}

	The proof of Theorem \ref{th4} is placed in Sec. \ref{pth4}.

\section{Applications of the resource theory}\label{s4}
In this section, we will first present the conversion rates of two states under asymptotic RNGs of a generic resource theory, then we consider the capacity of a classical communication task under the dynamical resource of coherence.

\subsection{Rates of state conversion under asymptotic RNGs}

In this subsection, we would consider a state transformation task for a generic resource theory $\langle\mathcal{F}_{R},\mathcal{O}_{R}\rangle$ defined in a system $\mathcal{H}.$ Here we assume the resource theory satisfies the conditions (R1)-(R4), furthermore, $\mathcal{F}_R=\mathcal{F}_R^{\tilde{T}}$.

Next we would 
recall the following quantifier $LR_G(\cdot)$, which is based on the generalized robustness $R_G(\cdot)$, 
\begin{align*}
		LR_G(\rho)=&log_2(1+R_G(\rho)).
\end{align*}
To illustrate the results of this section, we need also the following smoothed and regularized version of the above quantifiers,
\begin{align*}
	R_G^{\epsilon}(\rho)=&\inf_{\tilde{\rho}\in B_{\epsilon}(\rho)}R_G(\tilde{\rho}),\\
	LR_G^{\epsilon}(\rho)=&\inf_{\tilde{\rho}\in B_{\epsilon}(\rho)}LR_G(\tilde{\rho}),\\
	\overline{LR}_G(\rho)=&\sup_{\epsilon_n\rightarrow0}\limsup\limits_{n\rightarrow\infty}\frac{1}{n}LR_G^{\epsilon_n}(\rho^{\otimes n}).
\end{align*}
In the first and second formula, the infimum takes over all the states in the set $B_{\epsilon}(\rho)=\{\tilde{\rho}|\frac{1}{2}||\rho-\tilde{\rho}||_1\le \epsilon\}.$

\begin{definition} Assume $\mathcal{E}:\mathcal{H}_1\rightarrow\mathcal{H}_2$,  $\mathcal{E}$ is an $\epsilon$-RNG if there exists a quantum channel $\mathcal{G}\in \mathcal{M}_{R}$ such that 
\begin{align*}
	\frac{1}{2}||\mathcal{E}-\mathcal{G}||_{\diamond}\le \epsilon.
\end{align*}
\end{definition}

\begin{definition}\label{d2}
Assume $\langle\mathcal{F}_R,\mathcal{M}_R\rangle$ is a static resource which satisfies (R1)-(R4), {\color{black}let $\rho$ and $\phi$ be two states, we say $\phi$ can be converted into $\rho$ with rate, $E_{C,\phi}^{an}(\rho),$ via asymptotical RNGs is defined as follows,}
	\begin{widetext}
	\begin{align*}
		E_{C,\phi}^{an}(\rho)=\inf_{\{k_n,\epsilon_n\}}\{\limsup\limits_{n\rightarrow\infty}\frac{k_n}{n}:\lim\limits_{n\rightarrow\infty}\min_{\Lambda_n\in \epsilon_n-RNG}||\rho^{\otimes n}-\Lambda_n(\phi^{\otimes k_n})||=0,\lim\limits_{n\rightarrow\infty}\epsilon_n=0\},
	\end{align*}
	\end{widetext}
	here the infimum is taken over all sequences of integers $\{k_n\}$ and real numbers $\{\epsilon_n\}$.
\end{definition}
{\color{black}
\begin{theorem}\label{th6}
	Assume $\rho$ is a state in $\mathcal{D}_{\mathcal{H}}$, then we have
	\begin{align}
		E_{C,\phi}^{an}(\rho)\ge\frac{\overline{LR}_G(\rho)}{LR_G(\phi)}.
	\end{align}
\end{theorem}

The proof of Theorem \ref{th6} is placed is Sec. \ref{asy}.
\begin{remark}
{	\color{black} In the resource theory of bipartite entanglement, taking $\phi$ in Definition \ref{d2} to be the maximally entangled state yields the task of entanglement dilution, which has been studied since the inception of the quantum information theory  \cite{horodecki2009quantum,bennett1996,brandao2008,brandao2010reversible,brandao2009,lami2023,lami2025,hayashi2025the}. Efficient lower bound on entanglement cost are intimately related to the reversibility of entanglement manipulation.}

\end{remark}

}

\subsection{Classical communication in terms of coherence resource}

In this subsection, we consider the classical communication task under the resource of coherence \cite{wilde2013quantum}. Let $\{\ket{l}\}_{l=1}^d$ be the referenced basis of $\mathcal{H}$, if $\rho\in \mathcal{D}_{\mathcal{H}}$ can be presented as $$\rho=\sum_lp_l\ket{l}\bra{l},$$ then $\rho\in\mathcal{F}_{\mathbf{c}}$ is an incoherent state. $\mathcal{M}_{\mathbf{c}}$ denotes the set of all operations that transform incoherent states into some incoherent state,
\begin{align*}
	\mathcal{M}_{\mathbf{c}}=\{\mathcal{N}|\mathcal{D}\circ\mathcal{N}\circ\mathcal{D}=\mathcal{N}\circ\mathcal{D}\},
\end{align*}
here $\mathcal{D}(\cdot)=\sum_l \ket{l}\bra{l}\cdot\ket{l}\bra{l}$, $\{\ket{l}\}$ is the referenced basis of $\mathcal{H}.$

\begin{figure}[htbp]
	\centering
	\includegraphics[width=0.5\textwidth]{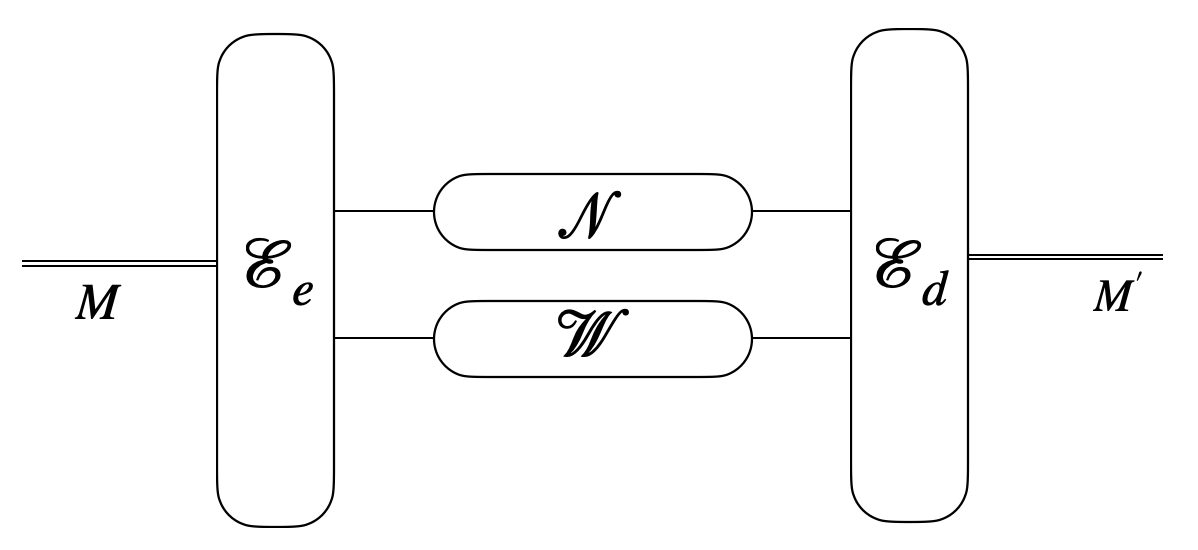}
	\caption{A protocol for classical communication over a quantum
		channel $\mathcal{N}$ with the aid of dynamical resource.}
	\label{fig1}
\end{figure}

Next we consider a classical communication task under a dynamical resource induced by coherence. The task played by Alice and Bob can be described as follows, 
\begin{itemize}
	\item[1.] Alice sends classical information to Bob with the aid of channels \begin{align}
		\Pi\circ\mathcal{N}=\mathcal{E}_d\circ(\mathcal{N}\otimes\mathcal{W})\circ\mathcal{E}_e,\label{pi}
	\end{align}
	here {\color{black}$\Pi(\cdot)=\mathcal{E}_d\circ(\cdot\otimes\mathcal{W})\circ\mathcal{E}_e$ is a superchannel,} $\mathcal{E}_d:\mathcal{H}_B\otimes\mathcal{H}_{ac}\rightarrow{M^{'}}$ and $\mathcal{E}_e:M\rightarrow\mathcal{H}_A\otimes\mathcal{H}_{ac}$ are free operations, and $\mathcal{W}$ is a channel in $\mathcal{M}_{\mathbf{c}}^{\tilde{T}}$.
	\item[2.]  Bob then performs a POVM $$\mathbb{M}=\{\ket{1}\bra{1},\cdots,\ket{m}\bra{m}\}$$ on the state after $\Pi\circ \mathcal{N}$. Here if the index $k\in \{1,2,\cdots,m\}$ happens, Bob concludes the classical message sent was $k$.
\end{itemize}
The average success probability to transmit $m$ messages through $\mathcal{N}$ with the aid of $\Pi$ is defined as 
\begin{align}
	f(\mathcal{N},m)=\sup_{\mathcal{E}_d,\mathcal{E}_e}\frac{\sum_k\mathrm{tr}(\Pi\circ\mathcal{N}(\ket{k}\bra{k})\ket{k}\bra{k})}{m},\label{fcm}
\end{align}
where the supremum of $\mathcal{E}_d:\mathcal{H}_B\otimes\mathcal{H}_{ac}\rightarrow{M^{'}}$ and $\mathcal{E}_e:M\rightarrow\mathcal{H}_A\otimes\mathcal{H}_{ac}$ takes over all the RNGs. To this end, the optimal error $\epsilon(\mathcal{N},m)$ can be written as 
\begin{align*}
	\epsilon(\mathcal{N},m)=1-f(\mathcal{N},m).
\end{align*}
The one-shot classical capacity with an error no more than $\theta\in (0,1)$ is defined as follows,
\begin{align*}
	c_{\theta}(\mathcal{N})=\log_2\sup\{m|\epsilon(\mathcal{N},m)\le \theta\}.
\end{align*}
{\color{black}	\begin{remark}
Actually, as $\mathcal{W}$ in (\ref{pi}) is in $\mathcal{M}_{\mathbf{c}}^{\tilde{T}},$ $\mathcal{I}\otimes\mathcal{W}$ is a RNG. As if $\mathcal{N}_1,\mathcal{N}_2\in \mathcal{M}_{\mathbf{c}},$ then \begin{align}
	&\mathcal{D}_{whole}\circ \mathcal{N}_1\otimes\mathcal{N}_2\circ\mathcal{D}_{whole}\nonumber\\=&\mathcal{D}_1\otimes\mathcal{D}_2\circ\mathcal{N}_1\otimes\mathcal{N}_2\circ\mathcal{D}_1\otimes\mathcal{D}_2]\nonumber\\=&\mathcal{N}_1\otimes\mathcal{N}_2\circ\mathcal{D}_1\otimes\mathcal{D}_2\nonumber\\=&\mathcal{N}_1\otimes\mathcal{N}_2\circ\mathcal{D}_{whole},
\end{align} 
that is, $\mathcal{N}_1\otimes\mathcal{N}_2\in \mathcal{M}_{\mathbf{c}}.$ Based on Remark \ref{re4}, $\mathcal{M}_{\mathbf{c}}=\mathcal{M}_{\mathbf{c}}^{\tilde{T}}$, then a generic superchannel $\Pi$ can be written as  $\Pi(\cdot)=\mathcal{E}_d\circ(\cdot\otimes\mathcal{I})\circ\mathcal{E}_e$, here $\mathcal{E}_d,\mathcal{E}_e\in \mathcal{M}_R.$
\end{remark}

An important class of the assisted communication are the ones where the sending and receiving party are limited to the no-signalling (NS) constraints \cite{PhysRevA.64.052309,PhysRevA.74.012305,ore2012,7115934,7353184,8012535,PhysRevLett.124.120502}. Assume $\Theta:(A^{'}\rightarrow B)\rightarrow(A\rightarrow B^{'})$ is realized by a bipartite NS channel $\Pi:AB\rightarrow A^{'}B^{'}$ satisfying:
\begin{align*}
	\mathrm{tr}_{A^{'}}\Pi(\rho_A\otimes\rho_B)=\mathrm{tr}_{A^{'}}\Pi(\sigma_A\otimes\rho_B),\\
	\mathrm{tr}_{B^{'}}\Pi(\rho_A\otimes\rho_B)=\mathrm{tr}_{B^{'}}\Pi(\rho_A\otimes\sigma_B).
\end{align*}
for any state $\rho_A,\rho_B$, $\sigma_A$ and $\sigma_B.$ Let $Y_{A_iB_iA_oB_o}$ be the Choi-Jamiolkowski matrix of the NS operation $\Pi$, then $Y$ satisfies \cite{7115934,8012535},
\begin{align}
	Y_{A_iB_iA_oB_o}\ge 0, \hspace{3mm} Y_{A_iB_i}=I_{A_iB_i},\label{p1}\\
	Y_{A_iB_iB_o}=\frac{I_{A_i}}{d_{A_i}}\otimes Y_{B_iB_o},\label{p2}\\
	Y_{A_iB_iA_o}=\frac{I_{B_i}}{d_{B_i}}\otimes Y_{A_iA_o}.\label{p3}
\end{align}
In this part, we would address on the classical communication through quantum channels assisted by NS and dynamical coherence nongenerating (DCNG) superchannels.}

{\color{black}

\begin{theorem}\label{th8}
	Assume $\mathcal{N}$ is a quantum channel, then 
	\begin{align}
		c_{\theta}(\mathcal{N})=&\log_2\sup m\\
		\textit{s. t.}\hspace{3mm}& \mathrm{tr} J_{\mathcal{N}}F\ge 1-\theta,\nonumber\\
		&0\le	F\le\rho_A\otimes I_B,\nonumber\\
		&	\mathrm{tr}_AF=\frac{I_B}{m},\nonumber\\
		&\mathrm{tr}\rho_A=1.\nonumber
	\end{align}

\end{theorem}

The proof of Theorem \ref{th8} is placed in \ref{pth8}. 

In the examples below, we use the CVX software \cite{grant} and
QETLAB  \cite{qetlab}
to solve the SDPs. 

\begin{Example}\label{e1}
	For the amplitude dampling channel, $$\mathcal{N}_{\gamma}(\cdot)=E_0\cdot E_0^{\dagger}+E_1\cdot E_1^{\dagger},$$ here $E_0=\ket{0}\bra{0}+\sqrt{1-\gamma}\ket{1}\bra{1}$ and $E_1=\sqrt{\gamma}\ket{0}\bra{1}.$ 
	
	In Fig. \ref{fig2}, we show the one-shot classical capacity of an amplitude dampling channel with an error $\theta$ when $p\in [0.5,0.98)$. The red line is with an error $\theta=0.01,$ and the blue line is with an error $\theta=0.004.$
	
	\begin{figure*}[htbp]
		\centering
		\includegraphics[width=0.5\textwidth]{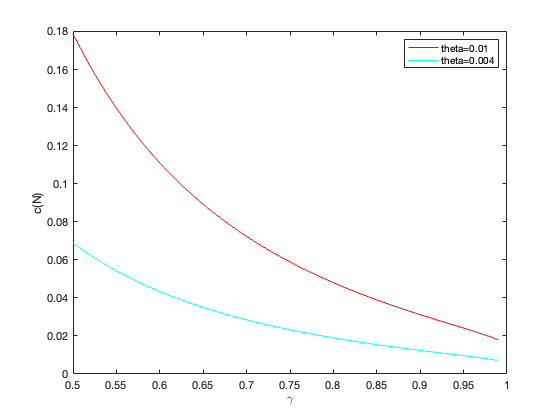}
		\large\caption{The one-shot classical capacity of an amplitude dampling channel. Here the red line is the one-shot classical capacity with an error $\theta=0.01,$ the blue line is the one-shot classical capacity with an error $\theta=0.004.$}
		\label{fig2}
	\end{figure*}
	\end{Example}

\begin{Example}\label{e2}
	For the BB84 channel \cite{PhysRevLett.100.170502}, $$\mathcal{N}_{\gamma}(\rho)=(1-\gamma)^2\rho+(\gamma-\gamma^2)(X\rho X+Z\rho Z)+\gamma^2Y\rho Y,$$ here 
	\begin{align*}
		X=\begin{pmatrix}
			0&1\\
			1&0
		\end{pmatrix},\hspace{3mm}Y=\begin{pmatrix}
		0&-i\\
		i&0
		\end{pmatrix},\hspace{3mm}Z=\begin{pmatrix}
		1&0\\
		0&-1
		\end{pmatrix}.
	\end{align*}
	
	In Fig. \ref{fig3}, we show the one-shot classical capacity of the BB84 channel with an error $\theta$. The capacity is not monotonic and reaches a minimum at $\gamma=0.5$ around with an error $\theta=0.03$ and $\theta=0.05$.
	
	\begin{figure*}[htbp]
		\centering
		\includegraphics[width=0.5\textwidth]{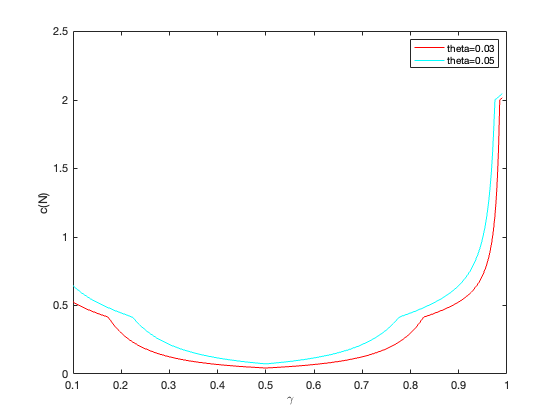}
		\large\caption{The one-shot classical capacity of the BB84 channel.}
		\label{fig3}
	\end{figure*}
\end{Example}

}
\section{Conclusion}
In this manuscript, we first considered RNGs on a generic static resource and presented a sufficient condition for the transformation between two states via a RNG within a generic resource theory. Next, we constructed a dynamical resource theory for RNGs corresponding to a given static resource theory, and we axiomatically proposed a method to quantify channels in the dynamical resource theory. We also studied the task of erasing a quantum channel under this resource theory. Finally, we demonstrated applications of RNGs in two quantum information tasks. Specifically, we presented lower bounds on the rates of  state conversion under asymptotic RNGs in the context of a generic resource theory, and we provided a bound on the classical capacity under the quantum dynamical resources of coherence. Given the fundamental importance of RNGs, we believe that our results will be valuable for further studies in quantum resource theories.

  \section{Acknowledgement}
X. S. was supported by the National Natural Science Foundation of China (Grant No. 12301580).

\section{Data availability}
The code that supports Example \ref{e1} and Example \ref{e2} of
this article are openly available \cite{sx}.  
\bibliographystyle{IEEEtran}
\bibliography{ref}
\setcounter{equation}{0}

\clearpage
\onecolumngrid
\setcounter{page}{1}
\setcounter{equation}{0}
\setcounter{figure}{0}
\renewcommand{\theequation}{S\arabic{equation}}
\renewcommand{\thefigure}{S\arabic{figure}}

\setcounter{secnumdepth}{1}

\section{Appendix}\label{app}

\subsection{Proof of Theorem \ref{th5}} \label{pth5}
Theorem \ref{th5}:\emph{	Assume $\langle\mathcal{F}_{R},\mathcal{M}_R\rangle$ satisfies (R1)-(R4), then the quantifiers $R_G(\cdot)$ and $\mathscr{R}_G(\cdot)$ for $\langle\mathcal{F}_R,\mathcal{M}_R\rangle$ satisfy the following properties,
	\begin{itemize}
		\item[(i)] For any mixed state $\rho$, $R_G(\rho),\mathscr{R}_G(\rho)\in[0,\infty).$
		\item[(ii)] $R_G(\cdot)$ and $\mathscr{R}_G(\cdot)$ are valid quantifiers in terms of $\langle\mathcal{F}_R,\mathcal{M}_R\rangle$.
		\item[(iii)] $R_G(\cdot)$ and $\mathscr{R}_G(\cdot)$ satisfy the properties O3 and O4.
		\item[(iv)] Assume $\rho$ is a state, $\sigma$ is the optimal free state in terms of $\mathcal{R}_G(\cdot)$ for $\rho$. Let $r=\mathcal{R}_G(\rho)$, if $s>r$,
		\begin{align*}
			\frac{\rho+s\sigma}{1+s}\in\mathcal{F}_{R}.
		\end{align*}
		Here $\mathcal{R}_G(\cdot)$ can be $\mathscr{R}_G(\cdot)$ or $R_G(\cdot).$
\end{itemize}}
\begin{proof} Here we only prove $\mathscr{R}_G(\cdot)$, and the proof of $R_G(\cdot)$ is similar.
	\begin{itemize}
		\item[(i)] As the resource is convex and compact, we could prove $\mathscr{R}_G(\rho)\in [0,\infty).$
		\item[(ii)] Here we only prove the property O2. Assume $\sigma$ is the optimal free state in terms of $\mathscr{R}_G(\cdot)$ for $\rho,$ then we have 
		\begin{align*}
			\frac{\rho+r\sigma}{1+r}\in \mathcal{F}_{R},\\
			\frac{\Lambda(\rho)+r\Lambda(\sigma)}{1+r}\in \mathcal{F}_{R},
		\end{align*}
		the second formula is due to that $\Lambda\in \mathcal{M}_R$. Due to the definition of $\mathscr{R}_G(\cdot)$, we have $\mathscr{R}_G(\Lambda(\rho))\le\mathscr{R}_G(\rho).$
		\item[(iii)] Here we only prove $\mathscr{R}_G(\cdot)$ satisfies the property O4. Assume $\mathscr{R}_G(\rho_1)=r_1,\mathscr{R}_G(\rho_1)=r_1$ and $\sigma_1$ and $\sigma_2$ are the optimal free states for $\rho_1$ and $\rho_2$ in terms of $\mathscr{R}_G(\cdot)$, respectively, 
		\begin{align*}
			\frac{p\rho_1+pr_1\sigma_1+(1-p)\rho_2+(1-p)r_2\sigma_2}{p(1+r_1)+(1-p)(1+r_2)}\in\mathcal{F}_{R},
		\end{align*}
		due to the definition of $\mathscr{R}_G(\cdot)$, we have $\mathscr{R}_G(p\rho_1+(1-p)\rho_2)\le p\mathscr{R}_G(\rho_1)+(1-p)\mathscr{R}_G(\rho_2).$
		\item[(iv)] 	When $\rho$ is a state, $\sigma$ is the optimal in terms of $\mathscr{R}_G(\cdot)$ for $\rho$. Let $r=\mathscr{R}_G(\rho)$, when $s>r$, 
		\begin{align*}
			\frac{\rho+s\sigma}{1+s}=\frac{1+r}{1+s}\times\frac{\rho+r\sigma}{1+r}+(1-\frac{1+r}{s+1})\sigma,
		\end{align*}
		as $\frac{\rho+r\sigma}{1+r}$ and $\sigma$ are free, and the set of free states is convex, $\frac{\rho+s\sigma}{1+s}$ is free.
	\end{itemize}
\end{proof}

\subsection{Proof of Theorem \ref{th3}}\label{pth3}
Theorem \ref{th3}:\emph{	Assume ${D}$ is a distance between two states, then
	\begin{itemize}
		\item if ${D}$ satisfies the data-processing property, 
		\begin{align*}
			D(\Lambda(\rho),\Lambda(\sigma))\le D(\rho,\sigma),\hspace{3mm}\forall \rho,\sigma\in \mathcal{D}_{\mathcal{H}}, \Lambda\in \mathcal{C}_{\mathcal{H}}.
		\end{align*}
		Then $\mathbb{F}_{D}$ is a RNG monotone. 
		\item if $\mathcal{E}_2\in \mathcal{M}_{R}^{\tilde{T}}$, then
		\begin{align*}
			\mathbb{F}_{D}(\mathcal{E}_1\otimes\mathcal{E}_2)=\mathbb{F}_{D}(\mathcal{E}_1)
		\end{align*}
		\item if $\mathcal{M}_{R}$ is closed under $D(\cdot,\cdot)$, $\mathbb{F}_D(\cdot)$ is faithful.
	\end{itemize}  }

\begin{proof}
	In this proof, we use $\mathcal{M}_{\mathcal{H}}^{\tilde{T}}$ to clarify the Hilbert space $\mathcal{H}$ that operations acts.
	Due to the definition of $D(\cdot,\cdot)$, the property P1 is satisfied by $\mathbb{F}_{D}$. Next we show the property P2, for a free super-channel $\mathcal{K}(\mathcal{E})=\mathcal{M}_2\circ(\mathcal{E}\otimes\mathcal{G})\circ\mathcal{M}_1$, $\mathcal{M}_1\in \mathcal{C}_{\mathcal{D}_A\rightarrow \mathcal{D}_{AS}}$ and $\mathcal{M}_2\in \mathcal{C}_{\mathcal{D}_{AS}\rightarrow\mathcal{D}_A}$ are RNGs, and $\mathcal{G}\in \mathcal{M}_{\mathcal{H}_S}^{\tilde{T}}$, then

		\begin{align*}
			&\mathbb{F}_D(\mathcal{K}(\mathcal{E}))\\
			=&\inf_{\Lambda\in \mathcal{M}_{\mathcal{H}_A}}\sup_{\rho_{AB},\Gamma\in \mathcal{M}_{\mathcal{H}_B}^{\tilde{T}}}D[(\mathcal{K}(\mathcal{E})\otimes\Gamma)(\rho_{AB}),(\Lambda\otimes\Gamma)(\rho_{AB})]\\
			\le&\inf_{\Pi\in \mathcal{M}_{\mathcal{H}_{AS}}}\sup_{\rho_{AB},\Gamma\in \mathcal{M}_{\mathcal{H}_B}^{\tilde{T}}}D[(\mathcal{K}(\mathcal{E})\otimes\Gamma)(\rho_{AB}),(\mathcal{M}_2\circ \Pi\circ\mathcal{M}_1\otimes\Gamma)(\rho_{AB})]\\
			\le &\inf_{\Pi\in \mathcal{M}_{\mathcal{H}_{AS}}}\sup_{\substack{\rho_{AB},\mathcal{G}\in\mathcal{M}^{\tilde{T}}_{\mathcal{H}_S},\\\Gamma\in \mathcal{M}_{\mathcal{H}^{\tilde{T}}_B}}}D[((\mathcal{E}\otimes\mathcal{G})\circ\mathcal{M}_1\otimes\Gamma)(\rho_{AB}),( \Pi\circ\mathcal{M}_1\otimes\Gamma)(\rho_{AB})]\\
			\le& \inf_{\Pi\in \mathcal{M}_{\mathcal{H}_{AS}}}\sup_{\substack{\rho_{ASB},\mathcal{G}\in\mathcal{M}^{\tilde{T}}_{\mathcal{H}_S},\\\Gamma\in \mathcal{M}_{\mathcal{H}_B}^{\tilde{T}}}}D[(\mathcal{E}\otimes\mathcal{G}\otimes \Gamma)(\rho_{ASB}),(\Pi\otimes\Gamma)(\rho_{ASB})]\\
			\le &\inf_{\Theta\in \mathcal{M}_{\mathcal{H}_{A}}}\sup_{\substack{\rho_{ASB},\mathcal{G}\in\mathcal{M}^{\tilde{T}}_{\mathcal{H}_S},\\\Gamma\in \mathcal{M}_{\mathcal{H}_B}^{\tilde{T}}}}D[(\mathcal{E}\otimes\mathcal{G}\otimes \Gamma)(\rho_{ASB}),(\Theta\otimes\mathcal{G}\otimes\Gamma)(\rho_{ASB})]\\
			\le&\inf_{\Theta\in \mathcal{M}_{\mathcal{H}_{A}}}\sup_{\rho_{AB},\Xi\in \mathcal{M}_{\mathcal{H}_B}^{\tilde{T}}}D[(\mathcal{E}\otimes \Xi)(\rho_{AB}),(\Theta\otimes\Xi)(\rho_{AB})]\\
			=&\mathbb{F}_D(\mathcal{E}),
		\end{align*}
	here the first inequality is due to that $\{\mathcal{M}_2\circ\Pi\circ\mathcal{M}_1|\Pi\in \mathcal{M}_{\mathcal{H}_{AS}}\}\subset \{\Lambda|\Lambda\in\mathcal{M}_{\mathcal{H}_A}\}$, the second inequality is due to the data-processing  property of $D(\cdot,\cdot)$, the fourth inequality is due to that $\{\Theta\otimes\mathcal{G}|\Theta\in \mathcal{M}_{\mathcal{H}_A}\}\subset\{\Pi|\Pi\in\mathcal{M}_{\mathcal{H}_{AS}}\}$, the last inequality is due to that $\mathcal{G}\otimes\Gamma$ in the penultimate is in $\mathcal{M}_R^{\tilde{T}}$.
	
	Next we show the following property, if $\mathcal{E}_1\in \mathcal{M}_{R}$ and $\mathcal{E}_2\in \mathcal{M}_{R}^{\tilde{T}}$, then

		\begin{align*}
			&\mathbb{F}_{D}(\mathcal{E}_1\otimes\mathcal{E}_2)\\
			=&\inf_{\Lambda\in \mathcal{M}_{\mathcal{H}_{AB_1}}}\sup_{\substack{\rho_{AB_1B_2},\\\Gamma\in \mathcal{M}_{\mathcal{H}_{B_2}}^{\tilde{T}}}}D[(\mathcal{E}_1\otimes\mathcal{E}_2\otimes\Gamma)(\rho_{AB_1B_2}),(\Lambda\otimes\Gamma)(\rho_{AB_1B_2})]\\
			\ge&\inf_{\Lambda\in \mathcal{M}_{\mathcal{H}_{AB_1}}}\sup_{\substack{\rho_{AB_2}\otimes\theta_{B_1},\\\Gamma\in \mathcal{M}_{\mathcal{H}_{B_2}}^{\tilde{T}}}}D[(\mathcal{E}_1\otimes\mathcal{E}_2\otimes\Gamma)(\rho_{AB_2}\otimes\theta_{B_1}),(\Lambda\otimes\Gamma)(\rho_{AB_2}\otimes\theta_{B_1})]\\
			\ge &\inf_{\Lambda\in \mathcal{M}_{\mathcal{H}_{AB_1}}}\sup_{\substack{\rho_{AB_2}\otimes\theta_{B_1},\\\Gamma\in \mathcal{M}_{\mathcal{H}_{B_2}}^{\tilde{T}}}}D[\mathrm{tr}_{B_1}(\mathcal{E}_1\otimes\mathcal{E}_2\otimes\Gamma)(\rho_{AB_2}\otimes\theta_{B_1}),\mathrm{tr}_{B_1}(\Lambda\otimes\Gamma)(\rho_{AB_2}\otimes\theta_{B_1})]\\
			\ge&\inf_{\Xi\in \mathcal{M}_{\mathcal{H}_{A}}}\sup_{\substack{\rho_{AB_2},\\\Gamma\in \mathcal{M}_{\mathcal{H}_{B_2}}^{\tilde{T}}}}D[(\mathcal{E}_1\otimes\Gamma)(\rho_{AB_2}),(\Xi\otimes\Gamma)(\rho_{AB_2 })]\\=&\mathbb{F}_D(\mathcal{E}_1),
		\end{align*} 
	here the first inequality is due to $\{\rho_{AB_1B_2}|\rho_{AB_1B_2}\in\mathcal{D}_{AB_1B_2}\}\supset\{\rho_{AB_1}\otimes\theta_{B_2}|\rho_{AB_1}\in\mathcal{D}_{AB_1},\theta_{B_2}\in\mathcal{D}_{{B_2}}\}$, the second inequality is due to that the partial trace is in $\mathcal{M}_{R}^{\tilde{T}}$, the last inequality is due to the assumption (R3). Then we show the other side of the inequality, 

		\begin{align*}
			&\mathbb{F}_{D}(\mathcal{E}_1\otimes\mathcal{E}_2)\\
			=&\inf_{\Lambda\in \mathcal{M}_{\mathcal{H}_{AB_1}}}\sup_{\rho_{AB_1B_2},\Gamma\in \mathcal{M}_{\mathcal{H}_{B_2}}^{\tilde{T}}}D[(\mathcal{E}_1\otimes\mathcal{E}_2\otimes\Gamma)(\rho_{AB_1B_2}),(\Lambda\otimes\Gamma)(\rho_{AB_1B_2})]\\
			\le& \inf_{\Lambda\in \mathcal{M}_{\mathcal{H}_{A}}}\sup_{\rho_{AB_1B_2},\Gamma\in \mathcal{M}_{\mathcal{H}_{B_2}}^{\tilde{T}}}D[(\mathcal{E}_1\otimes\mathcal{E}_2\otimes\Gamma)(\rho_{AB_1B_2}),(\Lambda\otimes\mathcal{E}_2\otimes\Gamma)(\rho_{AB_1B_2})]\\
			\le& \inf_{\Lambda\in \mathcal{M}_{\mathcal{H}_{A}}}\sup_{\rho_{AB_1},\Gamma\in \mathcal{M}_{\mathcal{H}_{B}}^{\tilde{T}}}D[(\mathcal{E}_1\otimes\Gamma)(\rho_{AB}),(\Lambda\otimes\Gamma)(\rho_{AB})]\\
			=&\mathbb{F}_D(\mathcal{E}_1).
		\end{align*}
	Here the last inequality is due to that $\mathcal{E}_2\otimes\Gamma\in \mathcal{M}^{\tilde{T}}.$ Hence, we finish the proof that $\mathbb{F}_D(\mathcal{E}_1\otimes\mathcal{E}_2)=\mathbb{F}_D(\mathcal{E}_1)$.

	At last, if $\mathcal{E}\in\mathcal{C}_{\mathcal{H}}$ satisfies $\mathbb{F}_D(\mathcal{E})=0,$ due to the definition of $\mathbb{F}_D(\cdot)$, we always could choose the ancillary system $\mathbb{C},$ there exists an operation $\Lambda_k\in \mathcal{M}_{R}$ such that 
	\begin{align*}
		\lim_k\sup_{\rho} D(\mathcal{E}(\rho),\Lambda_k(\rho))=0,
	\end{align*}
	as $\mathcal{M}_{R}$ is closed under the $D(\cdot,\cdot)$, then $\mathbb{F}_D(\cdot)$ is faithful.
\end{proof}

\subsection{Proof of Theorem \ref{th4}}\label{pth4}
	\begin{lemma}\label{l2}\cite{john2009}
	Assume $\Psi_1$ and $\Psi_2$ are two quantum channels, 
	\begin{align*}
		||\Psi_1-\Psi_2||_{\diamond}=\sup_{||A||_2=1}||(A\otimes \mathbb{I})(J_{\Psi_1}-J_{\Psi_2})||_1,
	\end{align*}
	here $J_{\Psi_1}$ and $J_{\Psi_2}$ are the Choi operators of $\Psi_1$ and $\Psi_2$, resepectively, and $||M||_1=\mathrm{tr}\sqrt{M^{\dagger}M}$.
\end{lemma}
\begin{lemma}\label{l10}
		Assume $\mathcal{Q}$, $\mathcal{N}$ and $\mathcal{N}^{'}$ are three quantum channels, and $\frac{1}{2}||\mathcal{N}-\mathcal{N}^{'}||_{\diamond}\le \delta$, then
		\begin{align}
			\frac{1}{2}||\mathcal{Q}\otimes(\mathcal{N}-\mathcal{N}^{'})||_{\diamond}\le \delta.
		\end{align}
\end{lemma}
\begin{proof}
	Assume $\mathcal{Q}(\cdot)$ can be written as $\mathcal{Q}(\sigma)=\mathrm{tr}_C[U(\sigma\otimes\gamma)U^{\dagger}]$,  $\gamma$ is a state, then
	\begin{align*}
			\frac{1}{2}||\mathcal{Q}\otimes(\mathcal{N}-\mathcal{N}^{'})||_{\diamond}=\max_{\rho}
		&\frac{1}{2}||(\mathcal{I}\otimes\mathcal{Q}\otimes\mathcal{I})[ (\mathcal{I}\otimes\mathcal{I}\otimes(\mathcal{N}-\mathcal{N}^{'}))(\rho)]||_{1}\\
		\le&\frac{1}{2}|| U(\varphi\otimes\gamma)U^{\dagger} ||_1\\
		=&\frac{1}{2}||\varphi\otimes\gamma||_1\le \delta,
	\end{align*}
	here $\varphi=\mathcal{I}\otimes\mathcal{I}\otimes(\mathcal{N}-\mathcal{N}^{'})(\rho)$. And the first inequality is due to the monotonicity under trace-preserving map \cite{Nielsen_Chuang_2010}. The last inequality is due to $	\frac{1}{2}||\mathcal{N}-\mathcal{N}^{'}||_{\diamond}\le \delta.$
\end{proof}
\begin{lemma}\label{l1}
	Assume $\mathsf{B}\sim B(n,p)$ is a binomially distributed random variable, its probability mass function is $f(n,p)=C_n^k p^k(1-p)^{n-k}$, when $k\le \floor{np},$
	\begin{align}
		f(n,p)\le\frac{e^{\frac{1}{12}}}{\sqrt{2\pi np(1-p)-2\pi}}.
	\end{align}
\end{lemma}
\begin{proof}
	\begin{align*}
		f(n,p)\le &\frac{\sqrt{2\pi n}(\frac{n}{e})^n e^{\frac{1}{12n}}}{\sqrt{2\pi k}(\frac{k}{e})^k\sqrt{2\pi(n-k)}(\frac{n-k}{e})^{n-k}}p^k(1-p)^k\nonumber\\
		\le&\frac{1}{\sqrt{2\pi}}\sqrt{\frac{n}{k(n-k)}}e^{\frac{1}{12n}}e^{n[\frac{k}{n}\ln\frac{pn}{k}+(1-\frac{k}{n})\ln\frac{1-p}{1-\frac{k}{n}}]}\\
		\le&\frac{1}{\sqrt{2\pi}}\sqrt{\frac{n}{k(n-k)}}e^{\frac{1}{12n}}\\
		\le& \frac{e^{\frac{1}{12}}}{\sqrt{2\pi np(1-p)-2\pi}}.
	\end{align*}
	Here the first inequality is due to the Stirling's formula, the third inequality is due to the nonnegativity of KL Divergence.
\end{proof}

\begin{lemma}\label{l3}
	Assume $\Psi$ and $\Phi$ are quantum channels in $\mathcal{C}_{\mathcal{H}}$, $p$ is in $(0,1)$ such that  $\Theta=p\Psi+(1-p)\Phi\in \mathcal{M}_{R}^{\tilde{T}}$. Let
	\begin{align*}
		\Gamma_n=\frac{1}{n}\sum_{i=1}^{n}\Theta^{\otimes(i-1)}\otimes \Psi\otimes\Theta^{n-i},
	\end{align*}
	For any $\epsilon>0,$ when $ n-1\ge \frac{2(1-p)e^{\frac{1}{6}}}{\epsilon^2\pi p}+\frac{1}{p(1-p)}$,
	\begin{align*}
		\frac{1}{2}	||\Gamma_n-\Theta^{\otimes n}||_{\diamond}\le \epsilon.
	\end{align*}
\end{lemma}
\begin{proof}
	Through computation, we have
	\begin{align*}
		\Theta^{\otimes n}-\Gamma_n=&(p\Psi+(1-p\Phi))^{\otimes n}-\frac{1}{n}\sum_{i=1}^{n}\Theta^{\otimes(i-1)}\otimes \Psi\otimes\Theta^{n-i}\nonumber\\
		=&\frac{1}{p}\sum_{k=0}^n C_n^k p^k(1-p)^{n-k}(p-\frac{k}{n})\Delta_k,
	\end{align*}
	here $\Delta_k=\frac{1}{C_n^k}\sum\limits_{|T|=k}\Psi^{T}\otimes\Phi^{[n]-T}$, $[n]=\{1,2,\cdots,n\}.$
	
	Then by Lemma \ref{l2}, we have
	\begin{align}
		||\Gamma_n-\Theta^{\otimes n}||_{\diamond}=&\sup_{||A||_2=1}||(A\otimes\mathbb{I})(J_{\Gamma_n}-J_{\Theta^{\otimes n}})||_1\nonumber\\
		\le&||A\otimes\mathbb{I}||_{\infty}||(J_{\Gamma_n}-J_{\Theta^{\otimes n}})||_{1}\nonumber\\
		\le&||(J_{\Gamma_n}-J_{\Theta^{\otimes n}})||_{1},\label{f0}
	\end{align}
	the first inequality is due to the H$\ddot{o}$ld inequality $||AB||_1\le ||A||_{\infty}||B||_1$, the second inequality is due to that $||A||_2=1.$
	Below we present the upper bound of $||J_{\Gamma_n-J_{\Theta^{\otimes n}}}||_1,$
	\begin{align}
		&||(J_{\Gamma_n}-J_{\Theta^{\otimes n}})||_{1}\nonumber\\
		\le &||(\mathbb{I}\otimes (\frac{1}{p}\sum_{k=0}^n C_n^k p^k(1-p)^{n-k}(p-\frac{k}{n})\Delta_k))(\ket{\psi_{+}}\bra{\psi_{+}})||_1\nonumber\\
		\le&\frac{1}{p}\sum_{k=0}^n C_n^k p^k(1-p)^{n-k}|p-\frac{k}{n}|,\label{fi1}
	\end{align}
	As 
	\begin{align*}
		&\sum_{k=0}^n C_n^k p^k(1-p)^{n-k}(p-\frac{k}{n})\\
		=&\sum_{k=0}^n C_n^k p^{k+1}(1-p)^{n-k}-\sum_{k=1}^nC_{n-1}^{k-1}p^k(1-p)^{n-k}\\
		=&p-p=0,
	\end{align*}
	in the first equality, we use $C_n^k\times\frac{k}{n}=C_{n-1}^{k-1}$. The second equality is due to that $(p+1-p)^m=\sum_{k=0}^mp^k(1-p)^{m-k}$. 
	
	Then (\ref{fi1}) turns into the following, 
	\begin{align}
		&\frac{1}{p}\sum_{k=0}^n C_n^k p^k(1-p)^{n-k}|p-\frac{k}{n}|\nonumber\\=&\frac{2}{p}\sum_{k=0}^{\floor{np}} C_n^k p^k(1-p)^{n-k}(p-\frac{k}{n})\nonumber\\
		=&\frac{2(1-p)^n}{p}\sum_{k=0}^{\floor{np}}[C_n^k\frac{p^{k+1}}{(1-p)^k}-C_{n-1}^{k-1}\frac{p^k}{(1-p)^k}]\nonumber\\
		=&\frac{2(1-p)^n}{p}\sum_{k=0}^{\floor{np}}[C_{n-1}^{k-1}\frac{p^k}{(1-p)^{k}}(p-1)+C_{n-1}^k\frac{p^{k+1}}{(1-p)^k}]\nonumber\\
		=&\frac{2}{p}\times(1-p)^n[C_{n-1}^{\floor{np}}\frac{p^{\floor{np}+1}}{(1-p)^k}-p]\nonumber\\
		\le& \frac{\sqrt{2}(1-p)e^{\frac{1}{12}}}{\sqrt{\pi p(n-1)(1-p)-\pi}}\label{sf4}
	\end{align}
	The third equality is due to that $C_n^m=C_{n-1}^m+C_{n-1}^{m-1},$ the last inequality is due to Lemma \ref{l1}.
	
	Combing (\ref{f0}), (\ref{fi1}), and (\ref{sf4}), when $n-1\ge \frac{2(1-p)e^{\frac{1}{6}}}{\epsilon^2\pi p}+\frac{1}{p(1-p)}$,
	\begin{align*}
		||\Gamma_n-\Theta^{\otimes n}||_{\diamond}\le \epsilon.
	\end{align*}
	
\end{proof}

\begin{lemma}\label{l13}\cite{liu2019resource}
		Assume $\Psi$ and $\Phi$ are quantum channels in $\mathcal{C}_{\mathcal{H}}$, $p$ is in $(0,1)$ such that  $\Theta=p\Psi+(1-p)\Phi\in \mathcal{M}_{R}^{\tilde{T}}$. Let
	\begin{align*}
		\Gamma_n=\frac{1}{n}\sum_{i=1}^{n}\Theta^{\otimes(i-1)}\otimes \Psi\otimes\Theta^{n-i},
	\end{align*}
	For any $\epsilon>0,$ when $ \log_2 n\ge \log_2\frac{1}{p}+\log_2\frac{1}{\epsilon^2}$,
	\begin{align*}
		\frac{1}{2}	||\Gamma_n-\Theta^{\otimes n}||_{\diamond}\le \epsilon.
	\end{align*}
\end{lemma}

Theorem \ref{th4}:\emph{
Assume $\mathcal{M}_R$ satisfies the following properties, 
\begin{itemize}
	\item $\mathcal{M}_R$ is closed under the permuation operation. That is, for a system $\mathcal{H}^{\otimes n}$, $\mathcal{U}_{\pi}(\cdot)=U^{\dagger}_{\pi}\cdot U_{\pi}\in \mathcal{M}_R$, $\pi$ is some permutation in $S_n$.
	\item $\mathcal{M}_R=\mathcal{M}_R^{\tilde{T}}.$
\end{itemize}
\noindent	
When	 $\mathcal{E}:A\longrightarrow B$ is a channel, then for arbitrary $\eta$ and $\epsilon$ such that $0<\eta<\epsilon<1$, we have
\begin{align*}
\mathbb{R}_L^{\sqrt{\epsilon}}(\mathcal{E})\le	C^{\epsilon}(\mathcal{E})\le \min\left\{\log_2\left[\frac{2(2^{\mathbb{R}_L^{\epsilon-\eta}(\mathcal{E})}-1)e^{\frac{1}{6}}}{\eta^2\pi }+\frac{1}{2^{-\mathbb{R}_L^{\epsilon-\eta}(\mathcal{E})}(1-2^{-\mathbb{R}_L^{\epsilon-\eta}(\mathcal{E})})}+1\right],\mathbb{R}_L^{\epsilon-\eta}(\mathcal{E})+\log_2\frac{1}{\eta^2}\right\}.
\end{align*} 
}

\begin{proof}
According to the definition of $\mathbb{R}_L^{\epsilon-\eta}(\cdot)$, we have there exists a channel $\mathcal{F}$ such that $\frac{1}{2}||\mathcal{E}-\mathcal{F}||_{\diamond}\le \epsilon-\eta$ and $\mathbb{R}_L^{\epsilon-\eta}(\mathcal{E})=\mathbb{R}_L(\mathcal{F})$. Let $-\log_2p=\mathbb{R}_L(\mathcal{F})$, $\mathcal{G}$ is the optimal channel such that $$\mathcal{K} =p\mathcal{F}+(1-p)\mathcal{G}\in \mathcal{M}_{R}^{\tilde{T}}.$$
Assume $\mathcal{U}_i$ is the permutation unitary channel between the first and the $i$-th subsystems, then 
\begin{align*}
	\mathcal{E}^{'}=&	\frac{1}{n}\sum_{i=1}^n \mathcal{U}_i\circ(\mathcal{K}\otimes\mathcal{E}^{\otimes(n-1)})\circ\mathcal{U}_i\\
	=&\frac{1}{n}\sum_{i=1}^n \mathcal{K}^{\otimes i-1}\otimes \mathcal{E}\otimes\mathcal{K}^{\otimes n-i},\\
	\mathcal{E}^{''}=&\frac{1}{n}\sum_{i=1}^n \mathcal{U}_i\circ(\mathcal{K}\otimes\mathcal{F}^{\otimes(n-1)})\circ\mathcal{U}_i\\
	=&\frac{1}{n}\sum_{i=1}^n \mathcal{K}^{\otimes i-1}\otimes \mathcal{F}\otimes\mathcal{K}^{\otimes n-i}
\end{align*}

Based on Lemma \ref{l10}, Lemma \ref{l3} and Lemma \ref{l13}, when taking $n\ge \min(\frac{2(1-p)e^{\frac{1}{6}}}{\eta^2\pi p}+\frac{1}{p(1-p)}+1,\frac{1}{p\eta^2})$, 
\begin{align*}
	||\mathcal{E}^{'}-\mathcal{K}^{\otimes n}||_{\diamond}\le& ||\mathcal{E}^{''}-\mathcal{E}^{'}||_{\diamond}+||\mathcal{E}^{''}-\mathcal{K}^{\otimes n}||_{\diamond}\\
	\le& 2(\epsilon-\eta+\eta)\\
	=&2\epsilon.
\end{align*}

Next we prove the lower bound of $C^{\epsilon}(\mathcal{E})$. Assume $\{\mathcal{U}_i,\mathcal{V}_i,p_i\}$ is a given $\epsilon$-destruction process of RNG for $\mathcal{E}$, and $\Lambda\in \mathcal{M}_{R}^{\tilde{T}}$, there exists $\mathcal{K}\in \mathcal{M}_{R}$,
\begin{align*}
	\frac{1}{2}||\sum_{i}^kp_i\mathcal{U}_i\circ(\mathcal{E}\otimes\Lambda)\circ\mathcal{V}_i-\mathcal{K}||_{\diamond}\le \epsilon.
\end{align*}

{\color{black}The proof of the lower bound is similar to that in \cite{liu2019resource,hsieh2020resource}. Assume there exists an ensemble of pairs of free reversible channels $\{p_i,\mathcal{U}_i,\mathcal{V}_i\}_{i=1}^k$ (i.e. $\mathcal{U}_i,\mathcal{V}_i\in \mathcal{M}_R, \forall i=1,2,\cdots,k),\mathcal{G}:A^{'}\longrightarrow B^{'}\in \mathcal{M}_R^{\tilde{T}}$, such that 
	\begin{align*}
		\frac{1}{2}||\sum_ip_i\mathcal{U}_i\circ(\mathcal{E}\otimes\mathcal{G})\circ\mathcal{V}_i-\Lambda||_{\diamond}\le \epsilon,
		\end{align*}
		let $\mathcal{N}_i=\mathcal{U}_i\circ(\mathcal{E}\otimes\mathcal{G})\circ\mathcal{V}_i$ and $\mathcal{M}=\Lambda.$
		By the relations between the trace norm and fidelity of states \cite{fuchs1999}, then 
		\begin{align*}
			P_{cb}(\sum_{i=1}^kp_i\mathcal{N}_i,\mathcal{M})\le \sqrt{\epsilon(2-\epsilon)},
		\end{align*}
		here 
		\begin{align*}
			P_{cb}(\mathcal{N},\mathcal{M})=\sqrt{1-\inf_{\rho} F^2((\mathcal{N}\otimes \mathcal{I})(\rho),(\mathcal{M}\otimes\mathcal{I})(\rho))}.
		\end{align*}
		
		Next let $W_i:AA^{'}\hookrightarrow BB^{'}\otimes E$ be the isometric dilations of the channels $\mathcal{N}_i$, we can write down an isometric dilation for $\overline{N}=\sum_ip_i\mathcal{N}_i$, $W=\sum_i\sqrt{p_i}W_i\otimes|i\rangle^F.$ By the Uhlmann theorem for the completely bounded fidelity \cite{dennis2008}, we can have an isometric dilation $Z=\sum_i\sqrt{p_i}Z_i\otimes|i\rangle^F$ of $\mathcal{M}$ that is $\sqrt{\epsilon(2-\epsilon)}$ close to $W$ with respect to $P_{cb}(\cdot,\cdot)$, Tracing out the subsystem $E$ and measuring $F$, and based on Fuchs-van de Graaf inequality \cite{fuchs1999,nielsen2002quantum}, we have
		
		\begin{align*}
					{\sqrt{\epsilon}}\ge& \frac{1}{2}||\sum_ip_i\mathcal{N}_i\otimes|i\rangle\langle i|-\sum_i p_i\mathcal{M}_i\otimes|i\rangle\langle i|||_{\diamond}\\
					=&\frac{1}{2}\sum_ip_i||\mathcal{N}_i-\mathcal{M}_i||_{\diamond}\\
					=&\frac{1}{2}\sum_ip_i||\mathcal{U}_i\circ(\mathcal{E}\otimes\mathcal{G})\circ\mathcal{V}_i-\mathcal{M}_i||_{\diamond}\\
					\ge&\frac{1}{2}	||\mathcal{E}\otimes\mathcal{G}-\sum_i p_i \mathcal{U}_i^{\dagger}\circ\mathcal{M}_i\circ\mathcal{V}_i^{\dagger}||_{\diamond}
		\end{align*}	}
where $\mathcal{M}_i=\mathrm{tr}_E W_i\cdot W_i^{\dagger}$ are completely positive with $\sum_ip_i\mathcal{M}_i=\mathcal{K}\in \mathcal{M}_R$, then  
\begin{align*}
	\sum_i p_i \mathcal{U}_i^{\dagger}\circ\mathcal{M}_i\circ\mathcal{V}_i^{\dagger}\le \sum_i \mathcal{U}_i^{\dagger}\circ\mathcal{K}\circ\mathcal{V}_i^{\dagger}.
\end{align*}
As the resource theory we considered here are convex, $\frac{1}{k}\sum_i\mathcal{U}_i^{\dagger}\circ\mathcal{M}_i\circ\mathcal{V}_i^{\dagger}\in \mathcal{M}_{R}^{\tilde{T}}$. Hence,
\begin{align}
	\mathbb{R}_L^{\sqrt{\epsilon}}(\mathcal{E}\otimes\mathcal{G})\le \mathbb{R}_L(\sum_i p_i\mathcal{U}_i^{\dagger}\circ\mathcal{M}_i\circ\mathcal{V}_i^{\dagger})\le \log_2k \label{sf0}
\end{align}

As $\mathcal{G}\in \mathcal{M_R^{\tilde{T}}}=\mathcal{M}_R$, we have
{\color{black}	\begin{align*}
\mathbb{R}_L^{\sqrt{\epsilon}}(\mathcal{E}_A\otimes\mathcal{G}_B)=&-\log_2\sup\{p|p\mathcal{K}_{AB}+(1-p)\mathcal{J}_{AB}\in \mathcal{M}_R\}\\
\ge&-\log_2\sup\{p|p\mathcal{K}_A+(1-p)\mathcal{J}_A\in\mathcal{M}_R\}\\
\ge&\mathbb{R}_L^{\sqrt{\epsilon}}(\mathcal{E}),
\end{align*}
here $\mathcal{K}_{AB}$ in the first equality is the optimal in terms of $\mathbb{R}_L^{\sqrt{\epsilon}}(\cdot)$ for $\mathcal{E}_A\otimes\mathcal{G}_B$. In the first inequality,
$\mathcal{L}_A(\cdot)=\mathrm{tr}_B\mathcal{L}_{AB}(\cdot\otimes\eta_B)$, here $\eta_B\in \mathcal{F}_R^{\tilde{T}},$ $\mathcal{L}$ is $\mathcal{K}$ or $\mathcal{J},$ the first inequality is due to that $\mathcal{M}_R$ is closed under partial trace and composition, the last inequality is due to the definition of $\mathbb{R}_L^{\epsilon}(\cdot)$ and that the diamond norm is monotone under the partial trace operation, then}
\begin{align}
	C^{\epsilon}(\mathcal{E})\ge  \mathbb{R}_L^{\sqrt{\epsilon}}(\mathcal{E}).
\end{align}
Hence we finish the proof.
\end{proof}

\subsection{Proof of Theorem \ref{th6}  }\label{asy}
\begin{lemma}\label{lRg}
	Assume $\Lambda\in \epsilon-\emph{RNG}$, let $\delta\ge \epsilon,$ then 
	\begin{align}
		R^{\delta}_G(\Lambda(\rho))\le& R_G(\rho),\\
		LR^{\delta}_G(\Lambda(\rho))\le& LR_G(\rho).
	\end{align}
\end{lemma}
\begin{proof}
	Assume $\rho$ is a state, $\sigma$ is the optimal in terms of $R_{G}(\cdot)$ for $\rho$, 
	\begin{align*}
		\rho+R_G(\rho)\sigma=(1+R_G(\rho))\pi,
	\end{align*}
	here $\pi\in\mathcal{F}_R$. Then we have
	\begin{align}
		\Lambda(\rho)+R_G(\rho)\Lambda(\sigma)=(1+R_G(\rho))\Lambda(\pi),
	\end{align}
	as $\Lambda$ is $\epsilon$-RNG, there exists a free channel $\Gamma\in \mathcal{M}_{R}$ such that 
	\begin{align}
	\frac{1}{2}	||\Lambda-\Gamma||_{\diamond}\le& \epsilon,\nonumber	\\
		\Gamma(\rho)+R_G(\rho)\Gamma(\sigma)=&(1+R_G(\rho))\Gamma(\pi),\label{r1}
	\end{align} 
	then 
	\begin{align*}
	\frac{1}{2}	||\Lambda(\rho)-\Gamma(\rho)||_1\le& \epsilon.
	\end{align*}
	That is, $\Gamma(\rho)\in B_{\epsilon}(\Lambda(\rho))$, combing $(\ref{r1})$, we have
	\begin{align}
		R^{\epsilon}_G(\Lambda(\rho))\le R_G(\rho),\nonumber\\
		LR^{\epsilon}_G(\Lambda(\rho))\le LR_G(\rho).
	\end{align}
	As $R_{G}^{\epsilon}(\cdot)$ and $LR_G^{\epsilon}(\cdot)$ are nonincreasing in terms of $\epsilon$, we finish the proof.
\end{proof}
\begin{lemma}\label{l4}
	Assume $\rho$ and $\sigma$ are two states, $\frac{1}{2}||\rho-\sigma||_1\le \epsilon$, then 
	\begin{align}
		R_G^{2\epsilon}(\sigma)\le R_G^{\epsilon}(\rho),\\
		LR_G^{2\epsilon}(\sigma)\le LR_G^{\epsilon}(\rho).
	\end{align} 
\end{lemma}
\begin{proof}
	Assume $\theta$ is the optimal for $\rho$ in terms of $R_G^{\epsilon}(\cdot)$, that is, 
	\begin{align*}
		\frac{1}{2}||\rho-\theta||_1\le \epsilon, \hspace{7mm}R_G^{\epsilon}(\rho)=R_G(\theta),
	\end{align*}
	then 
	\begin{align*}
		&\frac{1}{2}||\sigma-\theta||_1\\
		\le&\frac{1}{2}||\sigma-\rho||_1+\frac{1}{2}||\rho-\theta||_1\\
		\le&2\epsilon.
	\end{align*}
	That is, $\theta\in B_{2\epsilon}(\sigma).$	Then based on the definition of $R_G^{\epsilon}(\cdot)$, we have
	\begin{align*}
		R_G^{2\epsilon}(\sigma)\le R_G^{\epsilon}(\rho).
	\end{align*}
\end{proof}

\begin{lemma}\label{l5}
		Assume $\rho$ is a state, then
	\begin{align}
LR_G(\rho^{\otimes k})\le kLR_G(\rho).
	\end{align}
\end{lemma}
\begin{proof}
	Assume $\sigma$ is the optimal in terms of $R_G(\cdot)$ for $\rho$, let $r=R_G(\rho)$, let $\pi\in\mathcal{F}_R$ be the optimal state such that 
	\begin{align*}
		\rho+r\sigma=(1+r)\pi\in (1+r)\mathcal{F}_R,
	\end{align*}
	hence, 
	\begin{align*}
		&(\rho+r\sigma)^{\otimes k}\nonumber\\
		=&\rho^{\otimes k}+r\sum_{m=1}^k\sigma_{A_m}\otimes\rho_{\overline{A_m}}^{\otimes(k-1)}
		+\cdots+r^{k-1}\sum_{m=1}^k\rho_{A_m}\otimes\sigma_{\overline{A_m}}^{\otimes(k-1)}+r^k\sigma^{\otimes k}\\
		=&(1+r)^k\pi^{\otimes k},	\end{align*}
	as $\pi\in \mathcal{F}_R,$ $\pi^{\otimes k}\in \mathcal{F}_R$, then $R_G(\rho^{\otimes k})\le (r+1)^k-1,$
	that is,
	\begin{align*}
		1+R_G(\rho^{\otimes k})\le (1+r)^k,
	\end{align*}
	then 
	\begin{align*}
		LR_G(\rho^{\otimes k})\le kLR_G(\rho).
	\end{align*}
\end{proof}

Theorem \ref{th6}:\emph{
		Assume $\rho$ and $\phi$ are states, then we have
	\begin{align}
		E_{C,\phi}^{an}(\rho)\ge\frac{\overline{LR}_G(\rho)}{LR_G(\phi)}.
	\end{align}
}

\begin{proof}
	Assume $\Lambda_n\in \epsilon_n$-RNG is the optimal sequence of maps in terms of $E_{C,\phi}^{an}(\rho)$, there exists an integer $N_1$ such that when $n>N_1$, there exists $\Lambda_n\in RNG(\epsilon_n)$,  
	\begin{align*}
 	\frac{1}{2}	||\rho^{\otimes n}-\Lambda_n(\phi^{\otimes k_n})||\le \epsilon_n.
	\end{align*}
	Due to Lemma \ref{lRg}, Lemma \ref{l4} and Lemma \ref{l5}, we have
	\begin{align}
		\frac{1}{n}LR_G^{2\epsilon_n}(\rho^{\otimes n})\le&	\frac{1}{n}LR_G^{\epsilon_n}(\Lambda(\phi^{\otimes k_n}))\nonumber\\\le& \frac{1}{n}LR_G(\phi^{\otimes k_n})\nonumber\\
		\le&\frac{k_n}{n}LR_G(\phi),
	\end{align}
	then we have the following
	\begin{align}
		E^{an}_C(\rho)=&\limsup\limits_{n\rightarrow \infty}\frac{k_n}{n}\nonumber\\
		\ge& \frac{\sup\limits_{\epsilon_n\rightarrow0}\limsup\limits_{n\rightarrow\infty}\frac{1}{n}LR_G^{2\epsilon_n}(\rho^{\otimes n})}{LR_G(\phi)}\nonumber\\
		\ge&\frac{\overline{LR}_G(\rho)}{LR_G(\phi)}.
	\end{align}
\end{proof}

\subsection{Proof of Theorem \ref{th8}}\label{pth8}

{\color{black}
\begin{lemma}
	Assume $\mathcal{H}$ is a Hilbert space with finite dimensions, let $\{\ket{l}\}_{l}$ be the referenced basis of $\mathcal{H}$ and $\mathcal{D}_{\mathcal{H}}(\cdot)=\sum_l\ket{l}\bra{l}\cdot \ket{l}\bra{l}$.  If $\Pi$ is a superchannel which transforms $\mathcal{N}\in \mathcal{M}_R$ to some channel in $\mathcal{M}_R$, then $\Pi$'s Choi matrix satisfies 
	\begin{align}
		(\mathcal{P}\otimes \mathcal{K})Y=0,\label{p4}
	\end{align}
	here $\mathcal{K}=\mathcal{D}\otimes(I-\mathcal{D})$, $\mathcal{P}=I-\mathcal{K}$, and $Y$ is the C-J operator of $\Pi$.
\end{lemma}
\begin{proof}
	As $\mathcal{N}\in \mathcal{M}_R$, then 
	\begin{align*}
		&\mathcal{D}\circ\mathcal{N}\circ\mathcal{D}=\mathcal{N}\circ\mathcal{D}\\\Longleftrightarrow& (I-\mathcal{D})\circ\mathcal{N}\circ\mathcal{D}=0\\
		\Longleftrightarrow&(\mathcal{D}\otimes(I-\mathcal{D}))J_{\mathcal{N}}=0,
	\end{align*}
	The last formula is due to that $J_{\mathcal{A}\circ \mathcal{C}\circ\mathcal{B}}=(\mathcal{B}^T\otimes\mathcal{A})J_{\mathcal{C}}$ and $\mathcal{D}^T=\mathcal{D}.$ Next denote $\mathcal{K}=\mathcal{D}\otimes(I-\mathcal{D}),$ and $\mathcal{P}=I-\mathcal{K}.$
	
	At last, 
	\begin{align*}
		&	\Pi(\mathcal{N})\in \mathcal{M}_R,\hspace{2mm}\forall\mathcal{N}\in\mathcal{M}_R\\
		\Longleftrightarrow&\mathcal{K}\circ\Pi\circ\mathcal{P}=0\\
		\Longleftrightarrow &(\mathcal{P}\otimes \mathcal{K})Y=0.
	\end{align*}
\end{proof}

Theorem \ref{th8}:\emph{
Assume $\mathcal{N}$ is a quantum channel, then 
\begin{align*}
	c_{\theta}(\mathcal{N})=&\log_2\sup m\\
	\textit{s. t.}\hspace{3mm}& \mathrm{tr} J_{\mathcal{N}}F\ge 1-\theta,\\
	&0\le	F\le\rho_A\otimes I_B,\nonumber\\
	&	\mathrm{tr}_AF=\frac{I_B}{m},\nonumber\\
	&\mathrm{tr}\rho_A=1.\nonumber\\
\end{align*}
}

\begin{proof}

First we compute the following,
\begin{align}
	\sum_k \mathrm{tr}[\Pi\circ\mathcal{N}(\ket{k}\bra{k})\ket{k}\bra{k}]\nonumber=&\sum_k\mathrm{tr}[\mathcal{R}(\ket{k}\bra{k})\ket{k}\bra{k}]\nonumber\\
	=&\sum_k\mathrm{tr}J_{\mathcal{R}}\ket{kk}\bra{kk},\nonumber\\
	=&\mathrm{tr}J_{\mathcal{R}}D\label{fcm1}
\end{align}
here $\mathcal{R}=\Pi\circ\mathcal{N}$, $D=\sum_{k}\ket{kk}\bra{kk}$ and $J_{\mathcal{R}}=\mathrm{tr}_{AB}(J_{\mathcal{N}}^T\otimes\mathcal{I}_{MM^{'}})Y,$ here $Y$ is the Choi-Jamiolkowski matrix of $\Pi$.

As $\Pi$ considered here is NS and DCNG, then its Choi-Jamiolkowski matrix should satisfy properties (\ref{p1}), (\ref{p2}),(\ref{p3}) and (\ref{p4}).
And if $Y$ is a feasible NS and DCNG, then for any $\pi\in S_m$, here $S_m$ is the symmetric group of degree $m$, $Y^{'}=(\pi_M\otimes \pi_{M^{'}})Y(\pi_M\otimes\pi_{M^{'}})^{\dagger}$ is also feasible. Moreover, any convex combination of two feasible operators also is feasible. Therefore, if $Y$ is feasible, $\tilde{Y}$ is also feasible, which is defined as
\begin{align*}
	\tilde{Y}=\mathcal{T}(Y)=\frac{1}{m!}\sum_{\substack{\pi_M\in S_m,\\\pi_{M^{'}}\in S_m}}(\tau_{M}\otimes \tau_{M^{'}})Y(\tau_{M}\otimes \tau_{M^{'}})^{\dagger}.
\end{align*}

Next as $\mathcal{T}(D)=D$, we have
\begin{align*}
	\mathrm{tr}_{MM^{'}}[Y(\mathcal{I}\otimes D)]=&\mathrm{tr}_{MM^{'}}[Y(\mathcal{I}\otimes \mathcal{T}(D))]\\
	=\mathrm{tr}_{MM^{'}}[\tilde{Y}(\mathcal{I}\otimes D)],
\end{align*}

Based on \cite{7353184}, $\tilde{Y}$ can be written as
\begin{align}
	\tilde{Y}=F\otimes D+G\otimes (I-D),\label{fcmm}
\end{align}
where $F$ and $G$ are some operators. Due to the properties (\ref{p1}),(\ref{p2}),(\ref{p3}) and (\ref{p4}),
\begin{align}
	F\ge 0,G\ge 0,\nonumber\\
	F+(m-1)G=\rho_A\otimes I_B,\nonumber\\
	\mathrm{tr}_AF=\frac{I}{m},\mathrm{tr}\rho_A=1.\label{p0}
\end{align}

Combing (\ref{fcm1}) and (\ref{fcmm}), we have

\begin{align}
	f(\mathcal{N},m)=&\sup\mathrm{tr}(J_{\mathcal{N}}F),\\
	\textit{s. t.}\hspace{3mm}&0\le	F\le\rho_A\otimes I_B,\nonumber\\
	&	\mathrm{tr}_AF=\frac{I_B}{m},\nonumber\\
	&\mathrm{tr}\rho_A=1.\nonumber
\end{align}
\end{proof}
}

	\end{document}